\documentclass[runningheads]{llncs}

\usepackage{comment}
\includecomment{onlyonarxiv}
\excludecomment{onlyinproceedings}
\usepackage{algorithmicx}
\usepackage[noend]{algpseudocode}
\usepackage{algorithm}

\algnewcommand{\algorithmicswitch}{\textbf{switch}}
\algdef{SE}[SWITCH]{Switch}{EndSwitch}[1]{\algorithmicswitch\ #1\ \algorithmicdo}{\algorithmicend\ \algorithmicswitch}%
\algtext*{EndSwitch}%

\algnewcommand{\algorithmiccase}{\textbf{case}}
\algdef{SE}[CASE]{Case}{EndCase}[1]{\algorithmiccase\ #1}{\algorithmicend\ \algorithmiccase}%
\algtext*{EndCase}%

\algnewcommand{\algorithmicon}{\textbf{on}}
\algdef{SE}[ON]{On}{EndOn}[1]{\algorithmicon\ #1\ \algorithmicdo}{\algorithmicend\ \algorithmicon}%
\algtext*{EndOn}%

\algnewcommand{\algorithmicmacro}{\textbf{macro}}
\algdef{SE}[MACRO]{Macro}{EndMacro}[1]{\algorithmicmacro\ #1\ \algorithmicdo}{\algorithmicend\ \algorithmicmacro}%
\algtext*{EndMacro}%

\algrenewcommand{\algorithmicdo}{}
\algrenewcommand{\algorithmicthen}{}

\algnewcommand{\algorithmicgoto}{\textbf{goto}}%
\algnewcommand{\Goto}[1]{\algorithmicgoto~\ref{#1}}%

\algnewcommand{\algorithmicbreak}{\textbf{break}}%
\algnewcommand{\Break}[0]{\algorithmicbreak}%

\algnewcommand{\algorithmicwaiton}{\textbf{wait on}}%
\algnewcommand{\WaitOn}[1]{\algorithmicwaiton~{#1}}%

\PassOptionsToPackage{hyphens}{url}
\usepackage{url}
\usepackage{hyperref}
\hypersetup{breaklinks=true}

\usepackage[utf8]{inputenc}
\usepackage[T1]{fontenc}

\usepackage{amssymb,amsfonts,amsmath,amsthm}
\usepackage{graphicx}
\usepackage{xcolor}

\usepackage[caption=false,font=footnotesize]{subfig}
\usepackage{float}

\usepackage{IEEEtrantools}
\allowdisplaybreaks

\usepackage[shortlabels,inline]{enumitem}
\setlist[itemize]{leftmargin=5.5mm}
\setlist[enumerate]{leftmargin=5.5mm}

\usepackage{xstring}

\usepackage{datetime2}

\usepackage{siunitx}

\usepackage{varwidth}

\usepackage{tikz}
\usetikzlibrary{calc}
\usetikzlibrary{arrows}
\usetikzlibrary{arrows.meta}
\usetikzlibrary{patterns}
\usetikzlibrary{positioning}
\usetikzlibrary{decorations.pathreplacing}
\usetikzlibrary{shapes.misc}
\usetikzlibrary{spy}

\pgfdeclarelayer{bg1}
\pgfdeclarelayer{bg2}
\pgfsetlayers{bg1,bg2,main}

\usepackage{pgfplots}
\pgfplotsset{compat=1.14}
\usepgfplotslibrary{fillbetween}

\usepackage{multirow,booktabs}
\usepackage{makecell}

\usepackage{xspace}
\usepackage{ifthen}

\newcommand{\AAD}{Availability-Accountability Dilemma\xspace}

\newcommand{\aad}{availability-accountability dilemma\xspace}

\newcommand{\asr}{accountable safety resilience\xspace}
\newcommand{\Asr}{Accountable safety resilience\xspace}
\newcommand{\ASR}{Accountable Safety Resilience\xspace}

\newcommand{\ourprotocol}{our protocol\xspace}
\newcommand{\Ourprotocol}{Our protocol\xspace}

\newcommand{\chlc}[2]{\ensuremath{\mathsf{ch}_{#1}^{#2}}}

\newcommand{\LOGda}[2]{%
    \ifthenelse{\equal{#1}{}}{%
        \ensuremath{\mathsf{LOG}_{\mathrm{da}}^{#2}}%
    }{%
        \ensuremath{\mathsf{LOG}_{\mathrm{da},#1}^{#2}}%
    }%
}
\newcommand{\LOGbft}[2]{%
    \ifthenelse{\equal{#1}{}}{%
        \ensuremath{\mathsf{LOG}_{\mathrm{bft}}^{#2}}%
    }{%
        \ensuremath{\mathsf{LOG}_{\mathrm{bft},#1}^{#2}}%
    }%
}
\newcommand{\LOG}[2]{%
    \ensuremath{\mathsf{LOG}_{#1}^{#2}}
}
\newcommand{\LOGacc}[2]{%
    \ifthenelse{\equal{#1}{}}{%
        \ensuremath{\mathsf{LOG}_{\mathrm{acc}}^{#2}}%
    }{%
        \ensuremath{\mathsf{LOG}_{\mathrm{acc},#1}^{#2}}%
    }%
}
\newcommand{\tr}[2]{%
    \ifthenelse{\equal{#1}{}}{%
        \ensuremath{\mathsf{T}^{#2}}%
    }{%
        \ensuremath{\mathsf{T}_{#1}^{#2}}%
    }%
}
\newcommand{\wt}[2]{%
    \ifthenelse{\equal{#1}{}}{%
        \ensuremath{\mathsf{w}^{#2}}%
    }{%
        \ensuremath{\mathsf{w}_{#1}^{#2}}%
    }%
}

\newcommand{\LOGBLANKFIX}[0]{\ensuremath{\mathsf{LOG}}}

\newcommand{\PI}[0]{\ensuremath{\Pi}}
\newcommand{\PIlc}[0]{\ensuremath{\Pi_{\mathrm{lc}}}}
\newcommand{\PIbft}[0]{\ensuremath{\Pi_{\mathrm{bft}}}}
\newcommand{\PIacc}[0]{\ensuremath{\Pi_{\mathrm{acc}}}}

\newcommand{\Adv}[0]{\ensuremath{\mathcal A}}
\newcommand{\Env}[0]{\ensuremath{\mathcal Z}}

\newcommand{\ie}[0]{\emph{i.e.}\xspace}
\newcommand{\eg}[0]{\emph{e.g.}\xspace}
\newcommand{\cf}[0]{\emph{cf.}\xspace}

\newcommand{\Eg}[0]{\emph{E.g.}\xspace}

\newcommand{\GST}[0]{\ensuremath{\mathsf{GST}}}
\newcommand{\GAT}[0]{\ensuremath{\mathsf{GAT}}}

\newcommand{\tx}[0]{\ensuremath{\mathsf{tx}}}
\newcommand{\txs}[0]{\ensuremath{\mathsf{txs}}}
\newcommand{\negl}[0]{\ensuremath{\operatorname{negl}}}

\newcommand{\Tconfirm}[0]{\ensuremath{T_{\mathrm{confirm}}}}
\newcommand{\Tcheckpoint}[0]{\ensuremath{T_{\mathrm{cp}}}}
\newcommand{\Ttimeout}[0]{\ensuremath{T_{\mathrm{to}}}}
\newcommand{\Trecent}[0]{\ensuremath{T_{\mathrm{rec}}}}
\newcommand{\kcp}[0]{\ensuremath{k_{\mathrm{cp}}}}
\newcommand{\Thotstuff}[0]{\ensuremath{T_{\mathrm{hs}}}} %
\newcommand{\Tslot}[0]{\ensuremath{T_{\mathrm{slot}}}} %
\newcommand{\Cf}[0]{\ensuremath{\mathcal{C}}}
\newcommand{\Wf}[0]{\ensuremath{\mathcal{W}}}

\newcommand{\bprop}[1]{%
    \ifthenelse{\equal{#1}{}}{%
        \ensuremath{\Hat{b}}%
    }{%
        \ensuremath{\Hat{b}_{#1}}%
    }%
}

\newcommand{\ld}[1]{%
    \ifthenelse{\equal{#1}{}}{%
        \ensuremath{\mathrm{L}^{(c)}}%
    }{%
        \ensuremath{\mathrm{L}^{(#1)}}%
    }%
}

\newcommand{\adj}[0]{\ensuremath{\mathcal{J}}}

\newcommand{\betaA}[0]{\ensuremath{\beta_{\mathrm{a}}}}
\newcommand{\betaL}[0]{\ensuremath{\beta_{\mathrm{l}}}}

\newcommand{\CpReq}[3]{%
    \ifthenelse{\equal{#3}{}}{%
        \ensuremath{\langle\mathsf{#1},#2\rangle}%
    }{%
        \ensuremath{\langle\mathsf{#1},#2\rangle_{#3}}%
    }%
}

\newcommand{\clue}{evidence\xspace}
\newcommand{\clues}{evidences\xspace}

\newcommand{\Clues}{Evidences\xspace}

\newcommand{\AdvEnvSync}[0]{\ensuremath{(\Adv_{\mathrm{s}}, \Env_{\mathrm{s}})}}
\newcommand{\AdvEnvPsync}[0]{\ensuremath{(\Adv_{\mathrm{p}}, \Env_{\mathrm{p}})}}
\newcommand{\AdvEnvDA}[0]{\ensuremath{(\Adv_{\mathrm{da}}, \Env_{\mathrm{da}})}}
\newcommand{\AdvEnvPG}[0]{\ensuremath{(\Adv_{\mathrm{pda}}, \Env_{\mathrm{pda}})}}
\newcommand{\AdvEnvParameterized}[1]{\ensuremath{(\Adv_{\mathrm{#1}}, \Env_{\mathrm{#1}})}}

\theoremstyle{plain}
\newtheorem*{theorem*}{Theorem}

\definecolor{myParula01Blue}{RGB}{0,114,189}
\definecolor{myParula02Orange}{RGB}{217,83,25}
\definecolor{myParula03Yellow}{RGB}{237,177,32}
\definecolor{myParula04Purple}{RGB}{126,47,142}
\definecolor{myParula05Green}{RGB}{119,172,48}
\definecolor{myParula06LightBlue}{RGB}{77,190,238}
\definecolor{myParula07Red}{RGB}{162,20,47}

\tikzset{myparula11/.style={color=myParula01Blue,solid,mark=+,mark options={solid}}}
\tikzset{myparula12/.style={color=myParula01Blue,densely dashed,mark=x,mark options={solid}}}
\tikzset{myparula13/.style={color=myParula01Blue,densely dotted,mark=o,mark options={solid}}}
\tikzset{myparula14/.style={color=myParula01Blue,dashdotted,mark=triangle,mark options={solid}}}
\tikzset{myparula15/.style={color=myParula01Blue,dashdotdotted,mark=square,mark options={solid}}}

\tikzset{myparula21/.style={color=myParula02Orange,solid,mark=+,mark options={solid}}}
\tikzset{myparula22/.style={color=myParula02Orange,densely dashed,mark=x,mark options={solid}}}
\tikzset{myparula23/.style={color=myParula02Orange,densely dotted,mark=o,mark options={solid}}}
\tikzset{myparula24/.style={color=myParula02Orange,dashdotted,mark=triangle,mark options={solid}}}
\tikzset{myparula25/.style={color=myParula02Orange,dashdotdotted,mark=square,mark options={solid}}}

\tikzset{myparula31/.style={color=myParula03Yellow,solid,mark=+,mark options={solid}}}
\tikzset{myparula32/.style={color=myParula03Yellow,densely dashed,mark=x,mark options={solid}}}
\tikzset{myparula33/.style={color=myParula03Yellow,densely dotted,mark=o,mark options={solid}}}
\tikzset{myparula34/.style={color=myParula03Yellow,dashdotted,mark=triangle,mark options={solid}}}
\tikzset{myparula35/.style={color=myParula03Yellow,dashdotdotted,mark=square,mark options={solid}}}

\tikzset{myparula41/.style={color=myParula04Purple,solid,mark=+,mark options={solid}}}
\tikzset{myparula42/.style={color=myParula04Purple,densely dashed,mark=x,mark options={solid}}}
\tikzset{myparula43/.style={color=myParula04Purple,densely dotted,mark=o,mark options={solid}}}
\tikzset{myparula44/.style={color=myParula04Purple,dashdotted,mark=triangle,mark options={solid}}}
\tikzset{myparula45/.style={color=myParula04Purple,dashdotdotted,mark=square,mark options={solid}}}

\tikzset{myparula51/.style={color=myParula05Green,solid,mark=+,mark options={solid}}}
\tikzset{myparula52/.style={color=myParula05Green,densely dashed,mark=x,mark options={solid}}}
\tikzset{myparula53/.style={color=myParula05Green,densely dotted,mark=o,mark options={solid}}}
\tikzset{myparula54/.style={color=myParula05Green,dashdotted,mark=triangle,mark options={solid}}}
\tikzset{myparula55/.style={color=myParula05Green,dashdotdotted,mark=square,mark options={solid}}}

\tikzset{myparula61/.style={color=myParula06LightBlue,solid,mark=+,mark options={solid}}}
\tikzset{myparula62/.style={color=myParula06LightBlue,densely dashed,mark=x,mark options={solid}}}
\tikzset{myparula63/.style={color=myParula06LightBlue,densely dotted,mark=o,mark options={solid}}}
\tikzset{myparula64/.style={color=myParula06LightBlue,dashdotted,mark=triangle,mark options={solid}}}
\tikzset{myparula65/.style={color=myParula06LightBlue,dashdotdotted,mark=square,mark options={solid}}}

\tikzset{myparula71/.style={color=myParula07Red,solid,mark=+,mark options={solid}}}
\tikzset{myparula72/.style={color=myParula07Red,densely dashed,mark=x,mark options={solid}}}
\tikzset{myparula73/.style={color=myParula07Red,densely dotted,mark=o,mark options={solid}}}
\tikzset{myparula74/.style={color=myParula07Red,dashdotted,mark=triangle,mark options={solid}}}
\tikzset{myparula75/.style={color=myParula07Red,dashdotdotted,mark=square,mark options={solid}}}

\pgfplotsset{
    mysimpleplot/.style = {
        every axis plot/.prefix style={thick},
        width=1.0\linewidth,
        height=0.75\linewidth,
        title style={font=\scriptsize,align=center},
        legend cell align=left,
        legend style={font=\scriptsize},
        legend columns=3,
        legend style={
            at={(0.5,1)},
            yshift=0.3em,
            anchor=south,
            draw=none,
            /tikz/every even column/.append style={
                column sep=0.3em
            },
            cells={
                align=left
            }
        },
        grid=both,
        minor tick num=3,
        major grid style={solid,draw=gray!50},
        minor grid style={densely dotted,draw=gray!50},
        label style={font=\scriptsize,align=center},
        tick label style={font=\scriptsize},
    },
}

\usepackage{pifont}%
\newcommand{\cmark}{\ding{52}}%
\newcommand{\xmark}{\ding{56}}%

\title{The \AAD \\ and its Resolution via Accountability Gadgets}
\titlerunning{\AAD and Accountability Gadgets}

\author{Joachim Neu \and
Ertem Nusret Tas \and
David Tse}

\institute{Stanford University\\
\email{\{jneu,nusret,dntse\}@stanford.edu}%
\footnote{Extended version: \cite{neu2021aadarxiv}. The authors contributed equally and are listed alphabetically.}}

\begin{document}
\begin{onlyonarxiv}
\end{onlyonarxiv}
\begin{onlyinproceedings}
\end{onlyinproceedings}
\nocite{neu2021aadarxiv}
\maketitle

\begin{abstract}
For applications of Byzantine fault tolerant (BFT) consensus protocols where the participants are economic agents, recent works highlighted the importance of \emph{accountability}: the ability to identify participants who provably violate the protocol. At the same time, being able to reach consensus under dynamic levels of participation is desirable for censorship resistance. We identify an \emph{\aad}: in an environment with dynamic participation, no protocol can simultaneously be accountably-safe and live. We provide a resolution to this dilemma by constructing a
provably secure
optimally-resilient accountability gadget to checkpoint a longest chain protocol, such that the full ledger is live under dynamic participation and the checkpointed prefix ledger is accountable. Our
accountability gadget construction is black-box and can use any BFT protocol which is accountable under static participation. Using HotStuff as the black box, we implemented our construction as a protocol for the Ethereum 2.0 beacon chain, and our Internet-scale experiments with more than $4{,}000$ nodes show that the protocol 
achieves
the required scalability and has better latency than the 
current solution 
Gasper,
which was shown insecure by recent attacks.
\end{abstract}

\section{Introduction}
\label{sec:introduction}

\subsection{Accountability and Dynamic Participation}
\label{sec:intro-accountability}

Safety and liveness are the two fundamental security properties of consensus protocols. A protocol run by a distributed set of nodes is safe if the ledgers generated by the protocol are consistent across nodes and across time. It is live if all honest transactions eventually enter into the ledger.
Traditionally, consensus protocols are developed for fault-tolerant distributed computing, where a set of distributed computing devices aims to emulate a reliable centralized computer.
In modern decentralized applications such as cryptocurrencies,
consensus nodes are no longer just disinterested computing devices but are agents acting based on economic and other rationales. To provide the proper incentives to encourage nodes to follow the protocol, it is important that they can be held accountable for their protocol-violating behavior. %
This point of view is advocated by Buterin and Griffith \cite{casper} in the context of their effort to add accountability (among other things) to Ethereum's Proof-of-Work (PoW) longest chain protocol, and is also central to the design of Gasper \cite{gasper}, the protocol running Ethereum~2.0's Proof-of-Stake (PoS) beacon chain. In these protocols, accountability is used to incentivize proper behavior by slashing the stake of protocol-violating agents.
PoW protocols 
like
Bitcoin \cite{nakamoto_paper} or Ethereum 1.0 do not
assign identities to miners, and hence cannot be expected to provide accountability.
Even Nakamoto-style PoS protocols such as Cardano's Ouroboros family \cite{kiayias2017ouroboros,david2018ouroboros,badertscher2018ouroboros}
lack accountability.
On the other hand,
protocols that are designed to provide accountability include Polygraph \cite{CGG19} and Tendermint \cite{tendermint}, and a recent comprehensive work \cite{forensics} shows that accountability can be added on top of
many (but not all) `traditional' propose-and-vote-style Byzantine fault tolerant (BFT) protocols,
such as HotStuff \cite{yin2018hotstuff}, PBFT \cite{pbft}, or Streamlet \cite{streamlet,snapandchat}.
There is, however, another crucial difference between Nakamoto-style and propose-and-vote-style protocols.
While protocols from the first group
do not provide accountability,
they tolerate dynamic participation, a sought after feature of public permissionless blockchains not only for censorship resistance.
In Bitcoin, \eg,
the total hash rate varies over many orders of magnitude over the years. Yet, the blockchains remain continuously \emph{available}, \ie, live.
Protocols from the second group, oppositely,
provide accountability but do not tolerate dynamic participation.%
\footnote{%
For completeness, there are also protocols which neither provide accountability nor tolerate dynamic participation, \eg, Algorand \cite{algorand}.}
Why is there no protocol that both supports accountability and tolerates dynamic participation?

\subsection{\AAD and Resolution via Accountability Gadgets}
\label{sec:intro-aad}

Our first result says that it is impossible to support accountability for
\emph{dynamically available} protocols, \ie, protocols that are live under dynamic participation (\cf Theorem~\ref{thm:aa-dilemma}).
We call this the \emph{\aad}.

\label{sec:resolution-via-acc-gadgets}

Our second contribution is
to provide a resolution to the dilemma.
As no \emph{single} ledger protocol can simultaneously
be available and accountable,
we design and implement an accountability gadget which,
when applied to a longest chain protocol,
generates a dynamically available ledger $\LOGda{}{}$ and a checkpointed prefix ledger $\LOGacc{}{}$ with provably optimal security properties.

Consider a network with a total of $n$ permissioned nodes, and an environment where the network may partition and the nodes may go online and offline.
\begin{enumerate}
    \item (\textbf{P1: Accountability}) The accountable ledger $\LOGacc{}{}$ can provide an \asr of $n/3$ at all times (\ie, identify that many protocol violators in case of a safety violation), and it is live after
    a possible
    partition heals and greater than $2n/3$ honest nodes come online.
    \item (\textbf{P2: Dynamic Availability}) The available ledger $\LOGda{}{}$ is guaranteed to be safe after 
    a possible
    network partition and live at all times, provided that fewer than $1/2$ of the online nodes are adversarial.
\end{enumerate}

Note that while the checkpointed ledger is by definition always a prefix of the full available ledger, the above result says that the checkpointed ledger will catch up with the available ledger when the network heals and a sufficient number of honest nodes come online.
Users can choose
individually
whether to resolve the dilemma
in favor of
availability
or
accountability.
For example, under exceptional circumstances,
a coffee shop
might rather tolerate payments reverting than stalling,
while a car dealer
might prefer stalling over reverting payments.

The achieved resiliences are optimal, which
can be seen by comparing this result with \cite[Theorem B.1]{forensics} (for P1) and
\cite[Theorem 3]{pass2017} (for P2).
The checkpointed ledger $\LOGacc{}{}$ cannot achieve better \asr than $n/3$; it in fact achieves exactly that. The dynamically available ledger $\LOGda{}{}$ cannot achieve a better resilience than $1/2$; the ledger in fact achieves it.
Moreover, even if the network was synchronous at all times,
no protocol could have generated an accountable ledger with better resilience
(Theorem~\ref{thm:aa-dilemma}). So we are getting partition-tolerance for free, even though accountability is the goal.

\begin{figure}[t]
    \centering
    \begin{tikzpicture}[
        x=2cm,
        y=1.12cm,
    ]
        \scriptsize

        \node [align=center,draw,fill=white] (PIlc) at (0.65,0) {Checkpoint-\\respecting\\longest\\chain};
        \node [align=center,draw,rotate=90,fill=white] (cpvotegenerator) at (2,0) {Vote\\generator};
        \node [align=center,draw,rotate=90,fill=white] (PIbft) at (2.9,0) {Accountable\\consensus};
        \node [align=center,draw,rotate=90,fill=white] (cpvoteinterpreter) at (4,0) {Vote\\interpreter};
        \draw [Latex-] (PIlc.west) -- ++(-0.25,0) node [above] {$\txs$};
        \draw [-Latex] (PIlc) -- (cpvotegenerator);
        \draw [-Latex] (cpvotegenerator) -- (PIbft) node [midway,above,rotate=90,anchor=west] {Votes};
        \draw [-Latex] (PIbft) -- (cpvoteinterpreter) node [midway,above,rotate=90,anchor=west] {Votes};
        \draw [-Latex] (PIbft) -- (cpvoteinterpreter) node [midway,below] {$\LOGbft{}{}$};
        \draw [-Latex] (cpvoteinterpreter.south) -- ++(0.1,0) -- ++(0,0.9) -| (cpvotegenerator) node[pos=0.25,above] {Checkpoint decisions};
        
        \draw [-Latex] (cpvoteinterpreter.south) -- ++(0.3,0) -- ++(0,1.5) -| (PIlc) node[pos=0.25,above] {Checkpoint decisions};
        
        \node [rotate=90,anchor=west,align=left] at ($(PIlc.east)+(0.2,0)$) {Confirmed\\blocks};

        \begin{pgfonlayer}{bg1}
            \draw [fill=myParula06LightBlue!30,draw=none] ($(cpvotegenerator.north east)+(-0.2,0.75)$) rectangle ($(cpvoteinterpreter.south west)+(0.2,-0.65)$);
            \node [anchor=south east] at ($(cpvoteinterpreter.south west)+(0.2,-0.65)$) {\emph{Accountability gadget}};
        \end{pgfonlayer}
        \begin{pgfonlayer}{bg2}
            \node [circle,draw=black!60,text=black!80,fill=black!20,inner sep=1pt,xshift=-2pt,yshift=-2pt] at ($(PIlc.south west)$) {$\PIlc$};
            \node [circle,draw=black!60,text=black!80,fill=black!20,inner sep=1pt,xshift=-2pt,yshift=-2pt] at ($(PIbft.north west)$) {$\PIbft$};
            
            \draw [-Latex,myparula61] (cpvoteinterpreter.south) -- ++(0.75,0) node[above] {$\LOGacc{}{}$};
            \draw [-Latex,myparula21] (PIlc.east) -- ++(0.15,0) -- ++(0,-1.5) -- ($(cpvoteinterpreter.south)+(0.75,-1.5)$) node[above] {$\LOGda{}{}$};
        \end{pgfonlayer}
        \coordinate (c1) at ($(cpvoteinterpreter.south west)+(0.2,-0.65)$);
        \coordinate (c2) at ($(cpvotegenerator.north east)+(-0.2,0.75)$);
        \node [circle,text=myParula06LightBlue,draw=none,inner sep=1pt,fill=white,xshift=5pt,yshift=5pt] at (c1 -| c2) {$\PIacc$};
    \end{tikzpicture}
    \caption[]{%
        We construct an accountability gadget $\PIacc$
        from any accountable BFT protocol $\PIbft$
        and apply it to a longest-chain-type protocol $\PIlc$
        as follows:
        The fork choice rule of $\PIlc$ is modified
        to respect the latest checkpoint decision.
        Blocks confirmed by $\PIlc$
        are output as available ledger $\LOGda{}{}$.
        They are also the basis on which nodes
        generate a proposal and vote for the next checkpoint.
        To ensure that all nodes reach the same checkpoint decision, consensus is reached on which votes to count
        using $\PIbft$.
        Checkpoint decisions are output as accountable ledger $\LOGacc{}{}$ and fed back into the protocol
        to ensure consistency of future block production in $\PIlc$
        and future checkpoints with previous checkpoints.%
    }
    \label{fig:protocol-description-overall-diagram}
\end{figure}

The accountability gadget construction is shown in Figure~\ref{fig:protocol-description-overall-diagram}.
It is built on top of any existing longest chain protocol modified to respect the checkpoints. That is, new blocks are proposed and the ledger of confirmed transactions is determined based on the longest chain among all the chains containing the latest checkpointed block.
This gives the available full ledger $\LOGda{}{}$.
Periodically, 
nodes vote on the next checkpoint
(following a randomly selected leader's proposal).
To ensure that when tallying votes all nodes
base their decision for the next checkpoint
on the same set of votes,
any accountable BFT protocol designed for a fixed level of participation
can be used
(entirely as a black box)
to reach consensus on the votes.
The chain up to the latest checkpoint constitutes the accountable prefix
ledger $\LOGacc{}{}$.
The gadget ensures that block production and confirmation
in $\PIlc$ and future checkpoints honor established checkpoints.
When instantiated with an accountable BFT protocol
that is secure
under network partitions,
$\LOGacc{}{}$ inherits its partition-tolerance.

\begin{figure}[t]
    \centering
    \begin{tikzpicture}
        \begin{axis}[
            mysimpleplot,
            name=plot1,
            xlabel={Time [s]},
            ylabel={Ledger length\\~[blocks]},
            legend columns=4,
            xmin=0, xmax=2500,
            ymin=-1, ymax=250,
            height=0.3\linewidth,
            width=0.4\linewidth,
            x label style={at={(axis description cs:1.0,0.0)},anchor=south east},
        ]
        
            \addplot [myparula61,jump mark left,mark=none,ultra thick] table [x=t,y=LOGacc] {figures/prototype-experiments/ledgers-runtime2500-n4100-f0-Tcp300.txt};
            \label{leg:prototype-results-ledgers-nofaults-LOGacc}
            \addlegendentry{Accountable prefix ledger};
            
            \addplot [myparula21,jump mark left,mark=none] table [x=t,y=LOGava] {figures/prototype-experiments/ledgers-runtime2500-n4100-f0-Tcp300.txt};
            \label{leg:prototype-results-ledgers-nofaults-LOGava}
            \addlegendentry{Available full ledger};
            
            \legend{};
            
        \end{axis}
        \begin{axis}[
            mysimpleplot,
            name=plot2,
            xshift=0.29\linewidth,
            xlabel={Time [s]},
            legend columns=4,
            xmin=0, xmax=2500,
            ymin=-1, ymax=250,
            height=0.3\linewidth,
            width=0.4\linewidth,
            yticklabels={},
            legend columns=1,
            legend style={
                at={(1,1)},
                xshift=0.3em,
                anchor=north west,
                draw=none,
                /tikz/every even column/.append style={
                    column sep=0.3em
                },
                cells={
                    align=left
                }
            },
            x label style={at={(axis description cs:1.0,0.0)},anchor=south east},
        ]
        
            \addplot [myparula61,jump mark mid,mark=none,ultra thick] table [x=t,y=LOGacc] {figures/prototype-experiments/ledgers-runtime2500-n4100-f1025-Tcp300.txt};
            \label{leg:prototype-results-ledgers-faults-LOGacc}
            \addlegendentry{Accountable prefix ledger};
            
            \addplot [myparula21,jump mark mid,mark=none] table [x=t,y=LOGava] {figures/prototype-experiments/ledgers-runtime2500-n4100-f1025-Tcp300.txt};
            \label{leg:prototype-results-ledgers-faults-LOGava}
            \addlegendentry{Available full ledger};

        \end{axis}
    \end{tikzpicture}%
    \vspace{-0.3em}
    \caption[]{%
        \textbf{Left:}
        Ledger dynamics
        of a longest chain protocol
        outfitted with our accountability gadget
        based on HotStuff,
        measured
        with $4{,}100$ nodes distributed
        around the world.
        No attack.
        The available full ledger grows steadily.
        The accountable prefix periodically catches up
        whenever a new
        block is checkpointed.
        \textbf{Right:}
        Even in the presence of a $\beta=25\%$
        adversary
        who mines selfishly in $\PIlc$
        and boycotts leader duty
        in $\PIbft$ and $\PIacc$,
        $\LOGda{}{}$ grows steadily
        and $\LOGacc{}{}$ periodically
        catches up with $\LOGda{}{}$.
        Under attack,
        the growth rate of $\LOGda{}{}$
        is reduced (due to selfish mining)
        and $\LOGacc{}{}$'s catching up
        is occasionally slightly
        delayed due to leader timeouts.
        (Parameters
        $n=4100$,
        $\Tcheckpoint=5\,\mathrm{min}$,
        $\Ttimeout=1\,\mathrm{min}$,
        $\Thotstuff=20\,\mathrm{s}$,
        $\kcp=k=6$,
        all nodes online;
        \cf Sections~\ref{sec:protocol-description}, \ref{sec:experiments}%
        )%
    }
    \label{fig:prototype-results-ledgers-nofaults}
    \label{fig:prototype-results-ledgers-faults}
\end{figure}
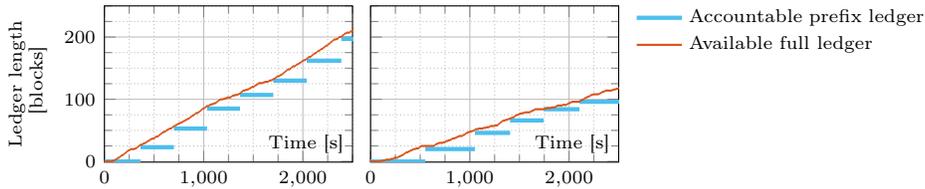

Since there are many accountable BFT  protocols \cite{forensics}, we have a lot of implementation choices.
Due to its
maturity and the availability of a high quality open-source implementation
which we could employ practically as a black box,
we decided to implement
a prototype of our accountability gadget using the HotStuff protocol \cite{yin2018hotstuff}. Taking the Ethereum 2.0's beacon chain as a target application
and matching its key performance characteristics
such as latency and block size,
we performed Internet-scale experiments to
demonstrate that our solution can meet the target specification
with over $4{,}000$ participants (see Figure~\ref{fig:prototype-results-ledgers-nofaults}(l)).
In particular,
for the chosen parameterization
and
even
before taking reduction measures,
the peak bandwidth required for a node to participate
does not exceed $1.5\,\mathrm{MB/s}$
(with a long-term
average of $78\,\mathrm{KB/s}$)
and hence 
is feasible even for many
consumer-grade Internet connections.
At the same time, our prototype provides
$5\times$
better average latency
of $\LOGacc{}{}$ compared to the
instantiation of Gasper
currently used
for Ethereum 2's beacon chain.

\subsection{Related Work}
\label{sec:related}

\subsubsection{Accountability}

Accountability in distributed protocols has been studied in earlier works.
\cite{HKD07} designed a system, PeerReview, which detects faults.
\cite{HK09} classifies faults into different types and studies their detectability. Casper \cite{casper} focuses on accountability and fault detection when there is violation of safety, and led to the notion of accountable safety resilience we use in this work. Polygraph \cite{CGG19} is a partially synchronous BFT protocol which is secure when there are less than $n/3$ adversarial nodes, and when there is a safety violation, at least $n/3$ nodes can be held accountable.
\cite{RG20} builds upon \cite{CGG19} to create a blockchain which can exclude Byzantine nodes that were found to have 
violated the protocol.

Many of these previous works focus on studying the accountability of \emph{specific} protocols and think of accountability as an add-on feature in addition to the basic security properties of the protocol. \cite{forensics} follows this  spirit but broadens the investigation to formulate a framework to study the accountability of many existing BFT protocols. More specifically, their framework augments the traditional resilience metric  with accountable safety resilience (which they call forensic support). The present work is more in the spirit of \cite{casper}
where accountability is a central design goal, not just an add-on feature. To formalize this spirit, we split traditional resilience into safety and liveness resiliences, upgrade safety resilience to accountable safety resilience, and formulate accountable security as a tradeoff between liveness resilience and accountable safety resilience. Further, we broaden the study to the important dynamic participation environment, where  we discovered the \aad (Theorem~\ref{thm:aa-dilemma}).
While at its heart the impossibility result Theorem~B.1 of \cite{forensics} is really about the tradeoff between liveness and accountable safety resiliences, although not stated as such, and it is indeed applicable very generally, when applied to the dynamic participation setting it would give a loose result and would not have been able to demonstrate the \aad.

\subsubsection{Availability-Finality Dilemma and Finality Gadgets}

\begin{table}[t]
    \centering
    \caption{%
    Accountability gadgets provide security, accountability,
    and predictable validity,
    which are not found conjoint in
    any one of the previous works
    \cite{gasper,snapandchat,ebbandflow,sankagiri_clc}.%
    }
    \setlength{\tabcolsep}{3pt}
    \begin{tabular}{lcccc}
        \toprule
            & \makecell{Gasper\\\cite{gasper}}
            & \makecell{Checkp.\ LC\\\cite{sankagiri_clc}}
            & \makecell{Snap\&Chat\\\cite{ebbandflow,snapandchat}}
            & \makecell{Acc. gadgets\\(This work)} \\
        \midrule
        Provable security 
            & \textcolor{myParula07Red}{\xmark{}}
            & \textcolor{myParula05Green}{\cmark{}}
            & \textcolor{myParula05Green}{\cmark{}}
            & \textcolor{myParula05Green}{\cmark{}}
            \\
        Accountability 
            & \textcolor{myParula05Green}{\cmark{}}
            & \textcolor{myParula07Red}{\xmark{}}
            & \textcolor{myParula05Green}{\cmark{}}
            & \textcolor{myParula05Green}{\cmark{}}
            \\
        Predictable validity 
            & \textcolor{myParula05Green}{\cmark{}}
            & \textcolor{myParula05Green}{\cmark{}}
            & \textcolor{myParula07Red}{\xmark{}}
            & \textcolor{myParula05Green}{\cmark{}}
            \\
        \bottomrule
    \end{tabular}
    \label{tab:comparison}
\end{table}

The \emph{avail\-ability-finality dilemma} \cite{PS_partition,lewispye2020resource,ebbandflow} states that no protocol can provide both finality, \ie, safety under network partitions, and availability, \ie, liveness under dynamic participation. The \emph{\aad} states that no protocol can provide both accountable safety and liveness under dynamic participation. Although they are different, it turns out that some, but not all, protocols that resolve the availability-finality dilemma can be used to resolve the \aad.
Casper \cite{casper} and Gasper \cite{gasper} pioneered resolution of the dilemmata
but lacked a specification of the desired security properties
and
suffered from attacks \cite{ethresearch-bouncing-attack-analysis,ethresearch-balancing-attack,ethresearch-balancing-attack2,ebbandflow,threeattacks,twoattacks}.
Specifically,
Gasper is insecure \cite{ethresearch-balancing-attack,ethresearch-balancing-attack2,ebbandflow,threeattacks}
(Table~\ref{tab:comparison}).
The first provably secure resolution of the availability-finality dilemma is the class of snap-and-chat protocols \cite{ebbandflow}, which combines a longest chain protocol with a partially synchronous BFT protocol in a black box manner to provide finality. If the partially synchronous BFT protocol is accountable, it is not too difficult to show \cite{snapandchat} that the resulting snap-and-chat protocol would also provide a resolution to the \aad. On the other hand, checkpointed longest chain \cite{sankagiri_clc}, another resolution of the availability-finality dilemma, is not accountable
(Table~\ref{tab:comparison}),
as shown in Appendix \ref{sec:appendix-non-accountability-clc}.

A strength of snap-and-chat protocols is its black box nature which gives it flexibility to provide additional features. A drawback is that the  protocol may reorder the blocks from the longest chain protocol to form the final ledger \cite{snapandchat}. This means that when a proposer proposes a block on the longest chain, it cannot predict the ledger state and check the  validity of transactions by just looking at the earlier blocks in the longest chain.  This lack of \emph{predictable validity} 
(Table~\ref{tab:comparison})
opens the protocol up to spamming and prohibits the use of standard techniques for sharding and light client support.
Checkpointed longest chain builds upon a line of work called finality gadgets \cite{casper,dinsdale2019afgjort,stewart2020grandpa,gasper} and overcomes this limitation of snap-and-chat protocols because the longest chain protocol is modified to respect the checkpoints so that the order of blocks can be preserved. However, checkpointed longest chain's finality gadget is not black box, but specifically uses Algorand BA \cite{algorand_agreement}, which is not accountable \cite{forensics}. It is not readily apparent if and how Algorand BA could be replaced with any accountable BFT protocol.

The accountability gadget we design combines structural elements from snap-and-chat protocols and from the checkpointed longest chain to uniquely achieve the best of both worlds.
It builds on the checkpointed longest chain and 
earlier (not provably secure)
finality gadgets
in that it complements a longest chain protocol with a checkpointing mechanism
and thus achieves predictable validity.
Like snap-and-chat protocols, it allows the use of any BFT protocol as a black box for checkpointing, retaining simplicity and flexibility and, when an accountable BFT protocol like HotStuff is used, the checkpointed ledger is accountable.
Our accountability gadget provides security, accountability,
and predictable validity (Table~\ref{tab:comparison}),
which are not found conjoint in any one of the prior works.%

\subsection{Outline}

We introduce in Section~\ref{sec:model} the notation and model for the proof of the \aad in Section~\ref{sec:availability-accountability-dilemma} and the construction and security proof of accountability gadgets in Section~\ref{sec:acc-gadgets-section}.
Finally, we discuss details of a prototype implementation
and
experimental performance results
in Section~\ref{sec:experiments}.

\section{Model}
\label{sec:model}

In the
client-server model of
state machine replication (SMR),
\emph{nodes} take inputs called \emph{transactions} and enable clients to agree on a single sequence of transactions, called the \emph{ledger} and denoted by $\LOG{}{}$, that produced the state evolution. 
For this purpose, nodes exchange messages, \eg, blocks or votes, and each node $i$ records its view of the protocol by time $t$ in an execution transcript $\tr{i}{t}$.
To obtain the ledger at time $t$, clients query the nodes running the protocol.
When a node $i$ is queried at time $t$, it produces
\emph{\clue} $\wt{i}{t}$ by applying an \emph{\clue generation function} $\Wf$ to its current transcript: $\wt{i}{t} \triangleq \Wf(\tr{i}{t})$.
Upon collecting \clues from some subset $S$ of the nodes, each client applies the \emph{confirmation rule} $\Cf$ to this set of \clues to obtain the ledger: $\LOG{}{} \triangleq \Cf(\{\wt{i}{t}\}_{i \in S})$.
Protocols typically require to query a subset $S$ containing at least one honest node.

\paragraph{Environment and Adversary:}
We assume that transactions are input to nodes by the environment $\Env$.
There exists a public-key infrastructure and each of the $n$ nodes is equipped with a unique cryptographic identity.
A random oracle serves as a common source of randomness.
Time is slotted and the nodes have synchronized
local
clocks.
\emph{Corruption:}
Adversary $\Adv$ is a probabilistic poly-time algorithm.
Before the protocol execution starts, $\Adv$ gets to corrupt (up to) $f$ nodes, then called \emph{adversarial} nodes. 
Adversarial nodes surrender their internal state to the adversary and can deviate from the protocol arbitrarily (Byzantine faults) under the adversary's control.
The remaining $(n-f)$ nodes are called \emph{honest} and follow the protocol as specified.
\emph{Networking:} Nodes can send each other messages.
Before a \emph{global stabilization time} $\GST$, $\Adv$ can delay network messages arbitrarily.
After $\GST$, $\Adv$ is required to deliver all messages sent between honest nodes within a known upper bound of $\Delta$ slots.
$\GST$ is chosen by $\Adv$, unknown to the honest nodes, and can be a causal function of the randomness in the protocol.
\emph{Sleeping:} To model dynamic participation, we adopt the concept of \emph{sleepiness} \cite{sleepy}.
Before a \emph{global awake time}%
\footnote{%
Node operators are rewarded and incur
little expenses for protocol participation.
Thus,
one naturally expects
frequent periods of (near) full participation.
$\GAT$ models the time when participation stabilizes,
analogous to the $\GST$ of network delays.%
} $\GAT$, $\Adv$ chooses, for every time slot and honest node, whether it is \emph{awake} (\ie, online) or \emph{asleep} (\ie, offline).
After $\GAT$, all honest nodes are awake.
An awake honest node 
executes the protocol faithfully.
An asleep honest node 
does not execute the protocol,
and messages that would have arrived in that slot are queued and delivered in the first slot in which the node is awake again.
Adversarial nodes are always awake.
We define $\beta$ as the maximum fraction of adversarial nodes among awake nodes throughout the execution of the protocol.
$\GAT$, just like $\GST$, is chosen by the adversary, unknown to the honest nodes and can be a causal function of the randomness.
But, while $\GST$ needs to happen eventually
($\GST < \infty$), $\GAT$ may be infinite.

Given above definition of a partially synchronous network with dynamic participation $\AdvEnvPG$, we model a synchronous network $\AdvEnvSync$, a partially synchronous network $\AdvEnvPsync$, and a synchronous network with dynamic participation $\AdvEnvDA$ as special cases with $\GST = \GAT = 0$, $\GAT = 0$, and $\GST = 0$, respectively.
Subsequently, we specify for every theorem
under which of the above four $\AdvEnvParameterized{...}$ it holds.
Examples of Nakamoto-style and propose-and-vote-style BFT protocols framed in the above model are given in Appendix~\ref{sec:appendix-examples-model}.

\paragraph{Safety and Liveness Resiliences:} Safety and liveness are defined as the traditional security properties of SMR protocols:
\begin{definition}
\label{def:security}
Let $\Tconfirm$ be a polynomial function of the security parameter $\sigma$ of an SMR protocol $\PI$.
We say that $\PI$ with a confirmation rule $\Cf$
is \emph{secure} and has transaction confirmation time $\Tconfirm$ if ledgers output by $\Cf$ satisfy:
\begin{itemize}
    \item \textbf{Safety:} For any time slots $t,t'$ and sets of nodes $S,S'$ satisfying the requirements stipulated by the protocol,
    either
    $\LOG{}{}\triangleq\Cf(\{\wt{i}{t}\}_{i \in S})$
    is a prefix of
    $\LOGBLANKFIX'\triangleq\Cf(\{\wt{i}{t'}\}_{i \in S'})$
    or vice versa.
    \item \textbf{Liveness:} If $\Env$ inputs a transaction to an awake honest node at some time $t$, then, for any time slot 
    $t' \geq \max(t,\GST,\GAT)+\Tconfirm$
    and any set of nodes $S$ satisfying the requirements stipulated by the protocol, the transaction is included in $\LOG{}{} \triangleq \Cf(\{\wt{i}{t'}\}_{i \in S})$.
\end{itemize}
\end{definition}

\begin{definition}
\label{def:safety}
For static (dynamic) participation, \emph{safety resilience} of a protocol is the 
maximum number $f$ of adversarial nodes (maximum fraction $\beta$ of adversarial nodes among awake nodes) such that the protocol satisfies safety.
Such a protocol provides \emph{$f$-safety} (\emph{$\beta$-safety}).
\end{definition}

\begin{definition}
\label{def:liveness}
For static (dynamic) participation, \emph{liveness resilience} of a protocol is the
maximum number $f$ of adversarial nodes (maximum fraction $\beta$ of adversarial nodes among awake nodes) such that the protocol satisfies liveness.
Such a protocol provides \emph{$f$-liveness} (\emph{$\beta$-liveness}).
\end{definition}

\paragraph{\ASR:} To formalize the concept of \asr, we define an \emph{adjudication function} $\adj$, similar to the forensic protocol defined in \cite{forensics}, as follows:
\begin{definition}
\label{def:adjudication-protocol}
An adjudication function $\adj$ takes as input two sets of \clues $W$ and $W'$ with conflicting ledgers $\LOG{}{}\triangleq\Cf(W)$ and $\LOGBLANKFIX'\triangleq\Cf(W')$, and outputs a set of nodes that have provably violated the protocol rules.
\end{definition}

So, $\adj$ never outputs an honest node.
When the clients observe a safety violation, \ie, at least two sets of \clues $W$ and $W'$ such that $\LOG{}{}\triangleq\Cf(W)$ and $\LOGBLANKFIX'\triangleq\Cf(W')$ conflict with each other, they call $\adj$ on these \clues to identify nodes that have violated the protocol.
Note that %
$\LOG{}{}\triangleq\Cf(\{\wt{i}{t}\}_{i \in S})$
may satisfy safety/liveness only
if the \clues come from
a set $S$ of nodes that satisfies some assumptions stipulated
by the protocol, \eg, that $S$ contains one honest
node.
On the other hand, $\adj$
should only output nodes that have undoubtedly violated protocol,
without
the verdict being conditional on any presumption.

\Asr builds on the concept of \emph{$\alpha$-ac\-count\-a\-ble-safety} first introduced in \cite{casper}:
\begin{definition}
\label{def:accountable-safety}
For static (dynamic) participation, \emph{\asr} of a protocol is the minimum number $f$ of nodes (minimum fraction $\beta$ of nodes among awake nodes) output by $\adj$
in the event of a safety violation.
Such a protocol provides \emph{$f$-accountable-safety} (\emph{$\beta$-accountable-safety}).
\end{definition}

Note that $\beta$-accountable-safety implies $\beta$-safety of the protocol (and the same for $f$) since $\adj$ outputs only adversarial nodes.

\section{The Availability-Accountability Dilemma}
\label{sec:availability-accountability-dilemma}

We observe that the strictest tradeoff between the liveness and \asr occurs for dynamically available protocols under $\AdvEnvDA$,
a result which was named the \aad in Section~\ref{sec:intro-aad}:
\begin{theorem}
\label{thm:aa-dilemma}
No SMR protocol provides both $\betaA$-accountable-safety and $\betaL$-liveness for any $\betaA,\betaL > 0$ under $\AdvEnvDA$.
\end{theorem}

Theorem~\ref{thm:aa-dilemma}
states that under dynamic participation it is impossible for an SMR protocol to provide both positive \asr and positive liveness resilience.
In light of this result, protocol designers
are compelled to choose between protocols that maintain liveness under fluctuating participation, and protocols that can enforce the desired incentive mechanisms highlighted in Section~\ref{sec:intro-accountability} via accountability.
Since both of the above features are desirable properties for Internet-scale consensus protocols, the \aad
presents a serious obstacle in the effort to obtain an incentive-compatible and robustly live protocol for applications such as cryptocurrencies.

To build intuition for the proof of Theorem~\ref{thm:aa-dilemma}, let us consider a permissioned longest chain protocol under $\AdvEnvDA$
where half of nodes are adversarial.
Adversarial nodes avoid all communication with honest nodes and build a private chain that conflicts with the chain built collectively by the honest nodes.
Such diverging chains mean the possibility of
an (ostensible) safety violation.
Think of an honest client towards whom adversarial nodes pretend to be asleep and who confirms a ledger based solely on the longest chain provided by the honest \clues;
and a co-conspirator of the adversary who pretends to not have received any \clues from honest nodes and
to have confirmed a ledger based solely on the longest chain provided by the adversarial \clues.
Indeed, both would obtain non-empty ledgers,
because the longest chain is dynamically available,
but these two ledgers would conflict.
Yet, based on the two sets of \clues,
the judge $\adj$ can
neither distinguish who is honest client and
who is co-conspirator,
nor tell which nodes are honest or adversarial.
So none of the adversarial nodes can be held accountable
(without risking to falsely convict an honest node).

Formal proof of Theorem~\ref{thm:aa-dilemma}
(Appendix~\ref{sec:appendix-aad-proof})
relies on the fact that in a dynamically available protocol, adversarial nodes, by private execution, can always create a set of \clues that yields a conflicting ledger through the confirmation rule $\Cf$.
This is because dynamically available protocols cannot set a lower bound on the number of \clues eligible to generate a non-empty ledger through $\Cf$, and thus are forced to output ledgers for \clues from any number of nodes.

Theorem~\ref{thm:aa-dilemma} is also related to a contemporaneous result \cite{lewispyeroughgardenccs} which shows that dynamically available protocols cannot produce certificates of confirmation, where such a certificate guarantees that there cannot be a conflicting confirmation so long as stipulated constraints on the adversary hold.

\section{Accountability Gadgets}
\label{sec:acc-gadgets-section}

In this section, we give a detailed description of the accountability gadget introduced in Section~\ref{sec:resolution-via-acc-gadgets}.
For ease of exposition, we construct
it
from 
an accountable BFT protocol $\PIbft$ 
with
accountable safety and liveness resilience of $\lfloor n/3 \rfloor$.

Like the checkpointed longest chain \cite{sankagiri_clc},
accountability gadgets %
output a prefix ledger safe under partial synchrony along with a full ledger live under dynamic participation.
For this purpose, both protocols
are deployed as overlays on top of
a dynamically available longest chain 
protocol
and periodically checkpoint its output to protect
against reversals under network partition.
Accountability gadgets can be instantiated from any partially synchronous BFT SMR protocol, which is used
as a black box for checkpointing.
If the selected protocol provides accountability,
then adversarial nodes can be held to account
should there ever be a reversal of a checkpoint.
In contrast,
the checkpointed longest chain is interwoven with a variant
of a particular protocol, Algorand BA \cite{algorand_agreement},
which does not provide accountability \cite{forensics} (\cf Appendix~\ref{sec:appendix-non-accountability-clc}).
Furthermore,
it is not readily apparent how to use another protocol instead.
As a result, the checkpointed longest chain cannot provide a resolution
to the \aad, whereas accountability gadgets can.

\subsection{Protocol Description}

\label{sec:protocol-description}

\begin{algorithm}[t]
    \caption{Checkpoint vote generator (helper functions: see Appendix~\ref{sec:appendix-pseudocode-helpers})}
    \label{algo:pseudocode-checkpoint-proposer}
    \begin{algorithmic}[1]
        \scriptsize
        \newcommand{\lastCp}{\mathrm{lastCp}}
        \newcommand{\proposals}{\mathrm{props}}
        \newcommand{\currIter}{\mathrm{currIter}}
        \newcommand{\me}{\mathrm{myself}}
        \State $\lastCp, \proposals \gets \bot, \{ c: \bot \mid c=0,1,... \}$
        \Comment{Last checkpoint, proposals}
        \For{$\currIter \gets 0,1,...$}
            \Comment{Loop over checkpoint iterations}
            \If{$\lastCp \neq \bot$}
                \While{waiting $\Tcheckpoint$ time}
                    \Comment{Wait $\Tcheckpoint$ time after new checkpoint decision}
                    \State $\Call{PerformBookkeeping}{}$
                \EndWhile
            \EndIf
            \If{$\operatorname{CpLeaderOfIter}(\currIter) = \me$}
                \Comment{Broadcast proposal if leader of current iteration}
                \State $\Call{Broadcast}{}(\CpReq{propose}{\currIter,\Call{GetCurrProposalTip}{\null}}{\me})$
                \label{line:leader-generator}
            \EndIf
            \While{waiting $\Ttimeout$ time}
                \Comment{Wait $\Ttimeout$ for timeout of checkpoint iteration}
                \State $\Call{PerformBookkeeping}{}$
                \On{$\proposals[\currIter] \neq \bot$, but at most once}
                    \Comment{Act on the first proposal received from authorized leader before end of $\Tcheckpoint$-wait and $\Ttimeout$-timeout}
                    \If{$\Call{IsValidProposal}{\proposals[\currIter]}$}
                    \label{line:isvalidproposal}
                        \Comment{Valid proposal is consistent with current checkpoint-respecting LC}
                        \State $\Call{SubmitVote}{}(\CpReq{accept}{\currIter,\proposals[\currIter]}{\me})$
                        \label{line:vote-accept-generator}
                    \Else
                        \State $\Call{SubmitVote}{}(\CpReq{reject}{\currIter}{\me})$
                        \label{line:vote-reject-generator}
                        \Comment{Reject invalid proposal}
                    \EndIf
                \EndOn
            \EndWhile
            \State $\Call{SubmitVote}{}(\CpReq{reject}{\currIter}{\me})$
            \label{line:vote-timeout-generator} \Comment{Reject due to timeout}
            \State \WaitOn{$\operatorname{Checkpoint}(c,b)$ \textbf{from} checkpoint vote interpreter (Alg.~\ref{algo:pseudocode-checkpoint-extractor}) \textbf{with} $c = \currIter$}
            \State $\lastCp \gets b$
            \label{loc:pseudocode-checkpoint-proposer-nextiter} \Comment{Keep track of checkpoint decision}
        \EndFor
        \Macro{$\Call{PerformBookkeeping}{}$}
            \On{receiving $\operatorname{Checkpoint}(c,b)$ \textbf{from} checkpoint vote interpreter (Alg.~\ref{algo:pseudocode-checkpoint-extractor}) \textbf{with} $c = \currIter$}
                \State \Goto{loc:pseudocode-checkpoint-proposer-nextiter}
                \label{line:goto2}
                \Comment{Jump to conclusion of current iteration}
            \EndOn
            \On{receiving $\operatorname{Proposal}(c,b)$ \textbf{from} checkpoint leader of iteration $c$ \textbf{with} $\proposals[c] = \bot$}
            \label{line:recvverifiedprop2}
                \State $\proposals[c] \gets b$
                    \Comment{Keep track of first proposal from authorized leader per iteration $c$}
            \EndOn
        \EndMacro
    \end{algorithmic}
\end{algorithm}

\begin{algorithm}[t]
    \caption{Checkpoint vote interpreter (helper functions: see Appendix~\ref{sec:appendix-pseudocode-helpers})}
    \label{algo:pseudocode-checkpoint-extractor}
    \begin{algorithmic}[1]
        \scriptsize
        \newcommand{\currIter}{\mathrm{currIter}}
        \newcommand{\currVotes}{\mathrm{currVotes}}
        \newcommand{\pk}{\mathrm{pk}}
        \newcommand{\Reject}{\mathrm{Reject}}
        \newcommand{\Accept}{\mathrm{Accept}}
        \newcommand{\committee}{\mathrm{committee}}
        \newcommand{\block}{\mathrm{block}}
        \newcommand{\req}{\mathrm{req}}
        \newcommand{\vote}{\mathrm{vote}}
        \For{$\currIter \gets 0,1,...$}
            \State $\currVotes \gets \{ (\pk, \bot) \mid \pk\in\committee \}$
            \Comment{Latest vote of each node}
            \While{true}
                \Comment{Go through votes as ordered by $\PIbft$}
                \State $\vote \gets \Call{GetNextVerifiedVoteFromBft}{\null}$
                \label{line:verifiedvote}
                \Comment{Verify signature}
                    \If{$\vote = \CpReq{accept}{c,b}{\pk}$ \textbf{with} $c = \currIter$}
                        \State $\currVotes[\pk] \gets \Accept(b)$
                        \Comment{Count accept vote for block $b$}
                    \ElsIf{$\vote = \CpReq{reject}{c}{\pk}$ \textbf{with} $c = \currIter$}
                        \State $\currVotes[\pk] \gets \Reject$
                        \Comment{Count reject vote}
                    \EndIf
                \If{$\exists b: |\{ \pk \mid \currVotes[\pk] = \Accept(b) \}| \geq 2n/3$}
                \label{line:threshold1}
                    \State $\Call{OutputCp}{\operatorname{Checkpoint}(\currIter, b)}$
                    \label{line:decision-accept-interpreter}
                    \Comment{New checkpoint decision}
                    \State \Break
                \ElsIf{$|\{ \pk \mid \currVotes[\pk] = \Reject \}| \geq n/3$}
                \label{line:threshold2}
                    \State $\Call{OutputCp}{\operatorname{Checkpoint}(\currIter,\bot)}$
                    \label{line:decision-reject-interpreter}
                    \Comment{Abort current iteration}
                    \State \Break
                \EndIf
            \EndWhile
        \EndFor
    \end{algorithmic}
\end{algorithm}

Accountability gadgets, denoted by $\PIacc$, can be used in conjunction with any dynamically available longest chain (LC) protocol $\PIlc$ such as Nakamoto's PoW LC protocol \cite{nakamoto_paper}, Sleepy \cite{sleepy}, Ouroboros \cite{kiayias2017ouroboros,david2018ouroboros,badertscher2018ouroboros} and Chia \cite{cohen2019chia} (Fig.~\ref{fig:protocol-description-overall-diagram}).
Subsequently, we focus on permissioned/PoS LC protocols.
PoW and Proof-of-Space are discussed in
Appendix~\ref{sec:appendix-extension-to-PoW}.
The protocol $\PIlc$ then follows a modified chain selection rule where honest nodes build on the tip of the LC that contains all of the \emph{checkpoints} they have observed.\footnote{There are no conflicting checkpoints unless a safety violation has already occurred.
Upon detecting a safety violation, honest nodes stop participating in the protocol. Punishment of parties identified by the accountability mechanism as malicious and system recovery are handled by mechanisms external to the protocol.}
We call
such chains 
\emph{checkpoint-respecting LCs}.
At each time slot $t$, each honest node $i$ outputs the $k$-deep prefix of the checkpoint-respecting LC (or the prefix of the latest checkpoint, whichever is longer) in its view as $\LOGda{i}{t}$.

The accountability gadget %
$\PIacc$ has three main components as shown on Fig.~\ref{fig:protocol-description-overall-diagram}:
a checkpoint vote generator (Alg.~\ref{algo:pseudocode-checkpoint-proposer}) 
issues checkpoint proposals and votes,
an accountable SMR protocol $\PIbft$ 
is used to reach consensus on which votes to count for the checkpoint decision,
and a checkpoint vote interpreter (Alg.~\ref{algo:pseudocode-checkpoint-extractor})
outputs checkpoint decisions computed deterministically from 
the checkpoint votes sequenced by $\PIbft$.
The protocol $\PIbft$ can be instantiated with any accountable BFT protocol,
\eg, 
Streamlet \cite{streamlet}, LibraBFT \cite{libraBFT}, or HotStuff \cite{yin2018hotstuff}.
It is used as a black box ordering service within $\PIacc$ and
is assumed to have confirmation time $\Tconfirm$.
We denote the ledger output by $\PIbft$ as $\LOGbft{}{}$,
and emphasize that it is internal to $\PIacc$.
Checkpoint vote generator and interpreter are run locally by each node and interact with $\PIbft$ and $\LOGbft{}{}$.
Hence, when
we refer to $\LOGbft{}{}$ in the following,
we mean the ledger in the view of a specific node.

The accountability gadget $\PIacc$ proceeds in \emph{checkpoint iterations} denoted by $c$, each of which attempts to checkpoint a 
block in $\PIlc$.
The checkpoint vote generator produces requests which can be of three forms:
$\CpReq{propose}{c,b}{i}$ 
proposes
block $b$ for checkpointing in 
iteration $c$, 
$\CpReq{accept}{c,b}{i}$ votes in favor of 
block $b$ in iteration $c$,
$\CpReq{reject}{c}{i}$ votes to reject
iteration $c$.
Here, $\langle ... \rangle_i$ denotes a message signed by node $i$.
Each iteration $c$ has a publicly verifiable and unique random leader $\ld{c}$.
The leader obtains the $\kcp$-deep block $b$ on its checkpoint-respecting LC
and broadcasts it to all other nodes as the checkpoint proposal for $c$
(Alg.~\ref{algo:pseudocode-checkpoint-proposer}, l.~\ref{line:leader-generator}).
Nodes receive checkpoint proposals
(signed by the legitimate leader $\ld{c}$)
from the network,
and order 
them with respect to their checkpoint iteration
(Alg.~\ref{algo:pseudocode-checkpoint-proposer},
l.~%
\ref{line:recvverifiedprop2}).
A proposal is
\emph{valid}
in view of node $i$
if the proposed block is within $i$'s checkpoint-respecting LC and extends all previous checkpoints observed by $i$.
During an iteration $c$, each node $i$ checks if the proposal received for $c$
is valid
(Alg.~\ref{algo:pseudocode-checkpoint-proposer}, l.~\ref{line:isvalidproposal}).
If it has 
received a valid proposal 
with 
block $b$, it votes $\CpReq{accept}{c,b}{i}$ (Alg.~\ref{algo:pseudocode-checkpoint-proposer}, l.~\ref{line:vote-accept-generator}).
Otherwise,
if $i$ does not receive
any valid proposal
for a timeout period $\Ttimeout$, $i$ votes $\CpReq{reject}{c}{i}$ (Alg.~\ref{algo:pseudocode-checkpoint-proposer}, l.~\ref{line:vote-reject-generator}, \ref{line:vote-timeout-generator}).
Votes are input 
as payload to $\PIbft$, which 
sequences them into
ledger $\LOGbft{}{}$.
Thus,
nodes
reach consensus on which votes to count for
checkpoint decision of the given iteration.

The checkpoint vote interpreter (Alg.~\ref{algo:pseudocode-checkpoint-extractor}) processes the
sequence of votes in 
$\LOGbft{}{}$ to produce checkpoint decisions.
Each node processes verified votes (\ie, with valid signature)
in the order they appear on
$\LOGbft{}{}$
(Alg.~\ref{algo:pseudocode-checkpoint-extractor}, l.~\ref{line:verifiedvote}).
Upon observing $2n/3$ unique $\CpReq{accept}{c,b}{i}$ votes
for
a block $b$ and the current iteration $c$,
each node outputs $b$ as the \emph{checkpoint} for 
$c$ 
(Alg.~\ref{algo:pseudocode-checkpoint-extractor}, l.~\ref{line:decision-accept-interpreter}).
The checkpointed blocks output
over time,
together with their respective prefixes, constitute
$\LOGacc{i}{t}$.
Furthermore, 
checkpoint decisions are fed back to $\PIlc$ and the checkpoint vote generator
to ensure consistency of future block production in $\PIlc$
and of checkpoint proposals with prior checkpoints.
Oppositely,
upon observing $n/3$ unique $\CpReq{reject}{c}{i}$ votes for the current iteration $c$,
each node outputs $\bot$ as the checkpoint decision for $c$
(Alg.~\ref{algo:pseudocode-checkpoint-extractor}, l.~\ref{line:decision-reject-interpreter}) to
signal that $c$
was aborted
with no new checkpointed block.
This
happens 
if honest nodes 
$\mathsf{reject}$
because they have not seen progress for too long.
Once a node outputs a 
decision for 
current iteration $c$,
the checkpoint vote interpreter 
proceeds to
$c+1$;
thus, only a single decision is output per iteration.

Upon receiving a new checkpoint for the current iteration $c$,
nodes leave 
$c$
of the checkpoint vote generator and enter
$c+1$
(Alg.~\ref{algo:pseudocode-checkpoint-proposer}, l.~%
\ref{line:goto2}).
If the checkpoint decision was for $b \neq \bot$,
nodes wait for $\Tcheckpoint$ time
(\emph{checkpoint interval})
before
considering
checkpoint proposals for $c+1$.
As will become clear in the analysis, the checkpoint interval is crucial to ensure that $\PIlc$'s chain dynamics are `not disturbed too much' by accommodating and respecting checkpoints.
Note that throughout the execution there is only a single
instantiation $\PIbft$, since the votes for different checkpoint
iterations can still be ordered into a single sequence.

\subsection{Security Properties}
\label{sec:proof-sketch}

\begin{figure*}[t]
    \hspace{-3cm}
    \begin{tikzpicture}[
        x=1cm,
        y=1.1cm,
        _box/.style = {
            align=center,
            text depth=5pt,
            draw,
        },
        _labelinline/.style = {
            align=center,
            anchor=south west,
            fill=black!10,
            draw,
        },
        statement/.style = {
            _box,
            labelinline/.style = {
                _labelinline,
            },
        },
        statementoutcome/.style = {
            _box,
            fill=black!5!green!10,
            labelinline/.style = {
                _labelinline,
                fill=white,
            },
        },
        leadsto/.style = {
            -Latex,
        },
    ]
        \tiny

        \begin{scope}
            \node (box1) at (0,-1) [statement] {\tikz{\node [labelinline] {1};} Consistency of check-\\pointed blocks in $\PIlc$};
            \node (box2) at (2.5,-1) [statement] {\tikz{\node [labelinline] {2};} Accountability\\ of $\LOGbft{}{}$};
            \node (box3) at (1,-2) [statementoutcome] {\tikz{\node [labelinline] {3};} Accountability \\of $\LOGacc{}{}$ (App.~\ref{sec:appendix-accountable-safety})};

            \draw [leadsto] (box1) -> (box3);
            \draw [leadsto] (box2) -> (box3);
        \end{scope}

        \begin{scope}[xshift=5.75cm]
            \node (box4) at (-1.5,0) [statement] {\tikz{\node [labelinline] {4};} Gap and recency properties\\of $\PIacc$ (App.~\ref{sec:appendix-gap-and-recency})};
            \node (box5) at (2.5,0) [statement] {\tikz{\node [labelinline] {5};} Recurrence of checkpoint-strong\\pivots (App.~\ref{sec:appendix-common-prefix})};
            \node (box6) at (0,-1) [statement] {\tikz{\node [labelinline] {6};} Security of $\PIlc$ after\\$\max(\GST,\GAT)$ (App.~\ref{sec:appendix-security-proof-for-the-longest-chain-protocol})};
            \node (box7) at (3.5,-1) [statement] {\tikz{\node [labelinline] {7};} Liveness of $\PIbft$ after\\$\max(\GST,\GAT)$};
            \node (box8) at (2,-2) [statementoutcome] {\tikz{\node [labelinline] {8};} Liveness of $\LOGacc{}{}$ after\\$\max(\GST,\GAT)$ (App.~\ref{sec:appendix-liveness})};

            \draw [leadsto] (box4) -> (box6);
            \draw [leadsto] (box5) -> (box6);
            \draw [leadsto] (box6) -> (box8);
            \draw [leadsto] (box7) -> (box8);
        \end{scope}

        \begin{scope}[xshift=12.6cm]
            \node (box9) at (0,-1) [statement] {\tikz{\node [labelinline] {9};} Checkpointing\\$\kcp$-deep blocks};
            \node (box10) at (2.4,-1) [statement] {\tikz{\node [labelinline] {10};} Security of $\PIlc$\\under synchrony};
            \node (box11) at (1.2,-2) [statementoutcome] {\tikz{\node [labelinline] {11};} Security of $\LOGda{}{}$};

            \draw [leadsto] (box9) -> (box11);
            \draw [leadsto] (box10) -> (box11);
        \end{scope}
    \end{tikzpicture}%
    \caption[]{Dependency of the security properties of $\LOGacc{}{}$ and $\LOGda{}{}$ on the properties of $\PIacc$, $\PIlc$ and $\PIbft$.}
    \label{fig:proof-sketch-diagram}
\end{figure*}

In this section, we formalize and prove the security properties \textbf{P1} and \textbf{P2} of Section~\ref{sec:resolution-via-acc-gadgets} %
for accountability gadgets based on \emph{permissioned} LC protocols \cite{sleepy,david2018ouroboros,kiayias2017ouroboros,badertscher2018ouroboros}.
(For an extension of the security analysis to Proof-of-Work and Proof-of-Space LC protocols, see Appendix~\ref{sec:appendix-extension-to-PoW}.)

For the worst case,
we first fix
$f = \lceil n/3 \rceil-1$
and consider an accountability gadget $\PIacc$ instantiated with a partially synchronous BFT protocol $\PIbft$ that provides $(n-2f)$-accountable-safety at all times, and $f$-liveness under partial synchrony after the network partition heals and sufficiently many honest nodes are awake.
(An example $\PIbft$
is HotStuff \cite{yin2018hotstuff} with a quorum size 
$(n-f)$.)
Let $\lambda$ and $\sigma$ denote the security parameters associated with the employed cryptographic primitives and the LC protocol $\PIlc$, respectively.
Then,
the security properties of 
$\LOGacc{}{}$ and $\LOGda{}{}$ output by the accountability gadget $\PIacc$ and the
LC protocol $\PIlc$ (modified to be checkpoint-respecting) are:
\begin{theorem}
\label{thm:main-security}
For any $\lambda, \sigma$,
and $\Tconfirm, k, \kcp$ linear in $\sigma$:
\begin{enumerate}
    \item (\textbf{P1: Accountability})
    Under $\AdvEnvPG$,
    the accountable ledger $\LOGacc{}{}$ provides $(n-2f)$-accountable-safety at all times (except with probability $\negl(\lambda)$),
    and there exists a constant $\mathbf{C}$ such that $\LOGacc{}{}$ provides $f$-liveness with confirmation time $\Tconfirm$ after $\mathbf{C}\max(\GST,\GAT)$ (except with probability $\negl(\sigma)$).
    \item (\textbf{P2: Dynamic Availability})
    Under $\AdvEnvDA$,
    the available ledger $\LOGda{}{}$ provides $1/2$-safety
    and $1/2$-liveness at all times
    (except with probability $\negl(\sigma)+\negl(\lambda)$).
    \item (\textbf{Prefix})
    $\LOGacc{}{}$ is always a prefix of $\LOGda{}{}$.
\end{enumerate}
\end{theorem}

Here,
$\negl(.)$ denotes a negligible function
that decays faster than all polynomials.
To prove Theorem~\ref{thm:main-security}, we first focus on the security of $\LOGda{}{}$ under $\AdvEnvDA$, 
synchronous network with dynamic availability
(\tikz[baseline=-2.75pt]{ \draw node[rectangle,draw,inner sep=2pt] {\tiny 11}; } of Fig.~\ref{fig:proof-sketch-diagram}).
We know from \cite{sleepy,david2018ouroboros,kiayias2017ouroboros,badertscher2018ouroboros} that $\PIlc$ is safe and live with some security parameter $\sigma$ under the original LC rule when $\beta<1/2$
(\tikz[baseline=-2.75pt]{ \draw node[rectangle,draw,inner sep=2pt,fill=black!10] {\tiny 10}; }).
Hence, if $\kcp$ is selected as an appropriate linear function of $\sigma$, once a block becomes $\kcp$-deep at time $s$ in the LC held by
an honest
node, it stays on the LCs held by all
honest nodes forever.
Since there are at least $n-f>f$ $\mathsf{accept}$ votes for any block checkpointed by an honest node at time $s$, there is at least one honest node that voted $\mathsf{accept}$ for any such block.
As honest nodes $\mathsf{accept}$ only proposals that are at least $\kcp$-deep in their LCs,
(\tikz[baseline=-2.75pt]{ \draw node[rectangle,draw,inner sep=2pt,fill=black!10] {\tiny 9}; }),
checkpointed blocks are already part of the LCs held by every other honest node at time $s$ under $\AdvEnvDA$.
Thus, new checkpoints can only appear in the common prefix of the honest nodes' LCs and do not affect the security of the LC protocol.
Next
accountability and liveness of $\LOGacc{}{}$ under $\AdvEnvPG$
(\tikz[baseline=-2.75pt]{ \draw node[rectangle,draw,inner sep=2pt] {\tiny 3}; }, \tikz[baseline=-2.75pt]{ \draw node[rectangle,draw,inner sep=2pt] {\tiny 8}; }).
The pseudocode of $\PIacc$ stipulates that honest nodes
$\mathsf{accept}$ only
proposals that are consistent with previous checkpoints
(\tikz[baseline=-2.75pt]{ \draw node[rectangle,draw,inner sep=2pt,fill=black!10] {\tiny 1}; }),
and a new checkpoint requires
$(n-f)$ $\mathsf{accept}$ votes (l.~\ref{line:threshold1} of Alg.~\ref{algo:pseudocode-checkpoint-extractor}).
Thus, in the event of a safety violation, either there are two inconsistent ledgers $\LOGbft{}{}$ held by honest nodes, or
$(n-2f)$ nodes
have voted for inconsistent checkpoints.
In both cases,
$(n-2f)$ adversarial nodes are identified as violators by invoking either
$(n-2f)$-accountable-safety of $\LOGbft{}{}$
(\tikz[baseline=-2.75pt]{ \draw node[rectangle,draw,inner sep=2pt,fill=black!10] {\tiny 2}; })
or the consistency requirement for checkpoints
(\tikz[baseline=-2.75pt]{ \draw node[rectangle,draw,inner sep=2pt,fill=black!10] {\tiny 1}; }),
implying $(n-2f)$-accountable-safety of $\LOGacc{}{}$.
Detailed proof 
in App.~\ref{sec:appendix-accountable-safety}.

Liveness of $\LOGacc{}{}$
(\tikz[baseline=-2.75pt]{ \draw node[rectangle,draw,inner sep=2pt] {\tiny 8}; })
requires the existence of iterations after $\max(\GST,\GAT)$ where all honest nodes 
$\mathsf{accept}$
honest proposals.
This, in turn, depends on whether the proposals by honest leaders are consistent with the checkpoint-respecting LCs 
at
honest nodes after $\max(\GST,\GAT)$.
To show this,
we prove that $\PIlc$ recovers its security after $\max(\GST,\GAT)$
(\tikz[baseline=-2.75pt]{ \draw node[rectangle,draw,inner sep=2pt,fill=black!10] {\tiny 6}; }).
We first observe that 
with
checkpoints, honest nodes abandon their LC if a new checkpoint 
appears 
on another (possibly shorter) chain.
Then, 
some honest blocks produced meanwhile may not
contribute to chain growth.
This feature of checkpoint-respecting LCs violates a core assumption of the standard proof techniques \cite{backbone,sleepy,kiayias2017ouroboros} for LC protocols.
To bound the number of 
abandoned honest blocks and demonstrate the \emph{self-healing}
property of checkpoint respecting LCs, we follow an
approach introduced in \cite{sankagiri_clc}. 
We first observe the \emph{gap} and \emph{recency} properties for $\PIacc$ (App.~\ref{sec:appendix-gap-and-recency})
which
are
necessary conditions for any checkpointing mechanism to ensure 
self-healing of $\PIlc$
(\tikz[baseline=-2.75pt]{ \draw node[rectangle,draw,inner sep=2pt,fill=black!10] {\tiny 4}; }).
The gap property states that 
$\Tcheckpoint$ has to be sufficiently longer than the time it takes for a
proposal to get checkpointed.
The recency property requires that newly checkpointed blocks were held in the checkpoint-respecting LC of at least one honest node within a short time interval
before
the checkpoint decision.

Using the gap and recency properties, we next extend the analysis of \cite{sankagiri_clc} to permissioned protocols by introducing the concept of \emph{checkpoint-strong pivots}, a generalization of strong pivots \cite{sleepy}.
Whereas strong pivots count honest and adversarial blocks to claim 
convergence of the LC in the view of different honest nodes, checkpoint-strong pivots consider only 
honest blocks that are guaranteed to extend the checkpoint-respecting LC, thus resolving
non-monotonicity for these chains.
Recurrence of checkpoint-strong pivots after $\max(\GST,\GAT)$
(\tikz[baseline=-2.75pt]{ \draw node[rectangle,draw,inner sep=2pt,fill=black!10] {\tiny 5}; })
along with the gap and recency properties
lead to
security of $\PIlc$ after $\max(\GST,\GAT)$. 
Details 
in App.~\ref{sec:appendix-security-proof-for-the-longest-chain-protocol}.
Given 
self-healing
of $\PIlc$, liveness of $\LOGacc{}{}$ follows from 
liveness of $\PIbft$ after $\max(\GST,\GAT)$
(\tikz[baseline=-2.75pt]{ \draw node[rectangle,draw,inner sep=2pt,fill=black!10] {\tiny 7}; }).
Full proof
in App.~\ref{sec:appendix-liveness}.

Finally, the prefix property follows readily from the way in which both $\LOGda{}{}$ and $\LOGacc{}{}$ are derived from the checkpoint-respecting LC.

\section{Experimental Evaluation}
\label{sec:experiments}

To evaluate whether the protocol of Section~\ref{sec:protocol-description}
can be
a drop-in replacement for 
the Ethereum 2 beacon chain,
we have implemented a prototype\footnote{Source code:
\url{https://github.com/tse-group/accountability-gadget-prototype}}.
\Ourprotocol incurs average required bandwidth comparable to Gasper at reduced latency of $\LOGacc{}{}$.
Gasper's resilience decreases
as the number of nodes increases, for fixed latency of 
$\LOGacc{}{}$,
due to a new attack %
\cite{ethresearch-balancing-attack2},
whereas our protocol is provably secure.
Supplemental
material of experimental evaluation is given in Appendix~\ref{sec:appendix-experiments}.

\begin{figure}[t]
    \centering
    \begin{tikzpicture}[
        x=2cm,
        y=1cm,
        component/.style = {
            minimum height=0.75cm,
            minimum width=2cm,
            align=center,
            draw,
            fill=white,
        },
        componentrotated/.style = {
            component,
            rotate=90,
        },
        arrowlabelrotated/.style ={
            midway,
            above,
            rotate=90,
            anchor=west,
            align=left,
        },
    ]
        \scriptsize

        \node [componentrotated,fill=black!10!blue!10] (lcminer) at (0,0) {LC block pro-\\duction lottery};
        \node [componentrotated,fill=black!10!green!10] (lcblocktree) at (1,0) {LC block\\tree manager};
        \node [componentrotated] (votegenerator) at (2,0) {Checkpoint\\vote generator};
        \node [componentrotated,fill=black!10] (hotstuff) at (3,0) {HotStuff};
        \node [componentrotated] (voteinterpreter) at (4,0) {Checkpoint\\vote interpreter};

        \node [component,minimum width=5cm,fill=black!10] (gossip) at (2,-2.5) {libp2p Gossipsub network};

        \node [black!30] (network) at (4.5,-2.5) {Network};
        \draw [Latex-Latex,black!30,dashed] (gossip.east) -- (network.west);

        \draw [-Latex] ($(lcblocktree.north)+(0,-0.5)$) -- ($(lcminer.south)+(0,-0.5)$) node [arrowlabelrotated] {Tip to grow};
        \draw [-Latex] (lcminer) |- (gossip) node [pos=0.37,right,align=left] {New\\blocks};
        \draw [-Latex] (gossip.north -| lcblocktree.west) -- (lcblocktree.west) node [midway,left,align=right] {New\\blocks};

        \draw [-Latex] ($(votegenerator.north)+(0,-0.425)$) -- ($(lcblocktree.south)+(0,-0.425)$) node [arrowlabelrotated] {Get/validate\\proposals};
        \draw [-Latex] ($(lcblocktree.south)+(0,-0.575)$) -- ($(votegenerator.north)+(0,-0.575)$);

        \draw [-Latex] ($(votegenerator.south)+(0,-0.5)$) -- ($(hotstuff.north)+(0,-0.5)$) node [arrowlabelrotated] {Votes};
        \draw [-Latex] ($(hotstuff.south)+(0,-0.5)$) -- ($(voteinterpreter.north)+(0,-0.5)$) node [arrowlabelrotated] {Votes};

        \draw [-Latex] ([xshift=-0.75mm] votegenerator.west) -- ([xshift=-0.75mm] votegenerator.west |- gossip.north);
        \draw [-Latex] ([xshift=0.75mm] votegenerator.west |- gossip.north) -- ([xshift=0.75mm] votegenerator.west) node [pos=0.75,right] {Proposals};

        \draw [-Latex] ([xshift=-0.75mm] hotstuff.west) -- ([xshift=-0.75mm] hotstuff.west |- gossip.north);
        \draw [-Latex] ([xshift=0.75mm] hotstuff.west |- gossip.north) -- ([xshift=0.75mm] hotstuff.west) node [pos=0.625,right,align=left] {HotStuff\\messages};

        \draw [-Latex,myparula21] ([yshift=-0.8cm] lcblocktree.south) -- (1.5,-1.8) -- (5.2,-1.8) node [above,xshift=-2em] {$\LOGda{}{}$};

        \coordinate (tmp1) at ($(voteinterpreter.south)+(0,-0.5)$);
        \coordinate (tmp2) at (5.2,0);
        \draw [-Latex,myparula61] ($(voteinterpreter.south)+(0,-0.5)$) -- (tmp1 -| tmp2) node [above,xshift=-2em] {$\LOGacc{}{}$};

        \draw [-Latex] ($(voteinterpreter.south)+(0,-0.5)$) -- ++(0.2,0) node [below,xshift=-3mm,anchor=north west] {Checkpoints} |- (3,1.4) -| (votegenerator);
        \draw [-Latex] ($(voteinterpreter.south)+(0,-0.5)$) -- ++(0.2,0) |- (3,1.4) node [above] {Checkpoints} -| (lcblocktree);

    \end{tikzpicture}%
    \caption{Components and their interactions in
    implementation
    of Fig.~\ref{fig:protocol-description-overall-diagram}.
    Gray: off the shelf components used as black boxes.
    Blue: taken from $\PIlc$ without modification.
    Green: taken from $\PIlc$, modified
    to respect checkpoints.}
    \label{fig:protocol-systems-diagram}
\end{figure}
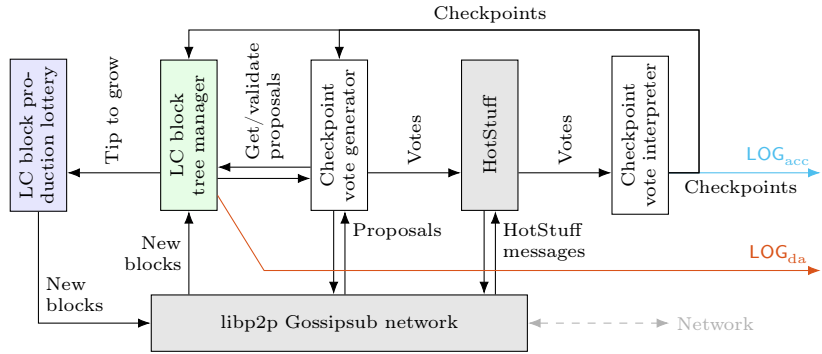

A diagram of the different components of our prototype
and their interactions is provided in Figure~\ref{fig:protocol-systems-diagram}.
We use a longest chain protocol modified to respect latest checkpoints as $\PIlc$, with a permissioned block production lottery;
and a variant of
HotStuff\footnote{We used this Rust implementation: \url{https://github.com/asonnino/hotstuff} \cite{hotstuff2chain}} as $\PIbft$.
All communication (including HotStuff's) takes place in a broadcast fashion via \texttt{libp2p}'s \texttt{Gossipsub} protocol\footnote{We used this Rust implementation: \url{https://github.com/libp2p/rust-libp2p} \cite{gossipsub}}, mimicking Ethereum 2 \cite{eth2-spec-p2p}.
The parameters of \ourprotocol 
match the number of validators ($n=4096$),
average block inter-arrival time ($12\,\mathrm{s}$)
and block payload size ($22\,\mathrm{KBytes}$)
of the
Ethereum 2 beacon chain.
We chose $\kcp=6$ so that an honest checkpoint proposal is
likely 
accepted by honest nodes,
and $k=6$ for swift $72\,\mathrm{s}$ average 
latency
of $\LOGda{}{}$.
Setting $\Thotstuff = 20\,\mathrm{s}$
and
$\Ttimeout = 1\,\mathrm{min}$
avoids HotStuff timeouts escalating into checkpoint timeouts unnecessarily.
Finally, to target $5\times$ improvement in average
$\LOGacc{}{}$ latency over Gasper (\cf Figure~\ref{fig:latency-bandwidth}), we set $\Tcheckpoint = 5\,\mathrm{min}$.

Adversarial nodes in the experiment
boycott
leader duty in $\PIbft$
and mine selfishly \cite{selfishmining} in $\PIlc$.
We ran our prototype
(a) with no adversary (Fig.~\ref{fig:prototype-results-ledgers-nofaults}(l)),
and (b) with $\beta=25\%$ adversary (Fig.~\ref{fig:prototype-results-ledgers-faults}(r)),
each 
for $2500\,\mathrm{s}$ on five AWS EC2 \texttt{c5a.8xlarge} instances in each of ten AWS regions
with $82$ nodes per machine, for a total of $4100$ nodes.
Each honest (adversarial) node connected to $15$ ($15$ honest, $15$ adversarial) randomly selected peers for the
peer-to-peer
network.
Both without
(Fig.~\ref{fig:prototype-results-ledgers-nofaults}(l))
and under
attack
(Fig.~\ref{fig:prototype-results-ledgers-faults}(r))
$\LOGda{}{}$
(\ref{leg:prototype-results-ledgers-nofaults-LOGava})
grows steadily, albeit under attack 
slower
due to selfish mining.
In both cases,
$\LOGacc{}{}$
(\ref{leg:prototype-results-ledgers-nofaults-LOGacc})
periodically catches up with $\LOGda{}{}$.
Timeouts cause
minor
delayed catch-up.

\begin{figure}[t]
    \centering
    \begin{tikzpicture}
        \begin{axis}[
            mysimpleplot,
            xlabel={Time [s] $\rightarrow$},
            ylabel={Rx bandwidth\\~[KB/s]},
            legend columns=4,
            xmin=900, xmax=1500,
            ymin=1e1, ymax=1e6,
            height=0.3\linewidth,
            width=0.85\linewidth,
            ymode=log,
            x label style={at={(axis description cs:1.0,1.0)},anchor=north east,xshift=0em},
        ]
        
            \draw [fill=red,draw=none,fill opacity=0.2] (axis cs:1010,2e1) rectangle (axis cs:1055,5e5);
            \draw [fill=red,draw=none,fill opacity=0.2] (axis cs:1340,2e1) rectangle (axis cs:1385,5e5);
            \draw [fill=orange,draw=none,fill opacity=0.2] (axis cs:1130,2e1) rectangle (axis cs:1190,3e4);
        
            \addplot [black,no marks,very thin] table [x=t,y=rx] {figures/prototype-experiments/traffic_runtime2500_n4100_f0_Tcp300_13.53.192.169.txt};
            \label{leg:prototype-results-traffic-nofaults-rx}

            \legend{};
            
        \end{axis}
    \end{tikzpicture}%
    \vspace{-0.3em}
    \caption[]{%
        Setting of Fig.~\ref{fig:prototype-results-ledgers-nofaults}(l):
        The network traffic
        for each AWS instance
        (\ie, $82$ nodes)
        shows four marked spikes (red)
        for every new checkpoint ($\Tcheckpoint=5\,\mathrm{min}$ interval) and smaller spikes (orange)
        for every new $\PIlc$ block ($\Tslot=7.5\,\mathrm{s}$ interval).
    }
    \label{fig:prototype-results-traffic-nofaults}
\end{figure}
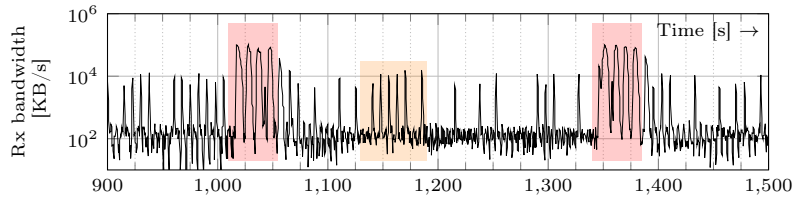

\begin{figure}[t]
    \centering
    \begin{tikzpicture}
        \begin{axis}[
            mysimpleplot,
            xlabel={Time [s] $\rightarrow$},
            ylabel={Rx bandwidth\\~[KB/s]},
            legend columns=4,
            xmin=850, xmax=1500,
            ymin=1e1, ymax=1e6,
            height=0.3\linewidth,
            width=0.55\linewidth,
            ymode=log,
            x label style={at={(axis description cs:0.5,1.0)},anchor=north west,xshift=0em},
        ]

            \draw [fill=red,draw=none,fill opacity=0.2] (axis cs:927,2e1) rectangle (axis cs:1077,7e5);
            \draw [fill=orange,draw=none,fill opacity=0.2] (axis cs:1195,2e1) rectangle (axis cs:1205,1e5);
            \draw [fill=orange,draw=none,fill opacity=0.2] (axis cs:1136,2e1) rectangle (axis cs:1146,1e5);
        
            \addplot [black,no marks,very thin] table [x=t,y=rx] {figures/prototype-experiments/traffic_runtime2500_n4100_f1025_Tcp300_3.238.227.12.txt};
            \label{leg:prototype-results-traffic-faults-rx-adversary}

            \legend{};
            
        \end{axis}
        \begin{axis}[
            mysimpleplot,
            xshift=0.45\linewidth,
            xlabel={Time [s] $\rightarrow$},
            legend columns=4,
            xmin=850, xmax=1500,
            ymin=1e1, ymax=1e6,
            height=0.3\linewidth,
            width=0.55\linewidth,
            yticklabels={},
            ymode=log,
            x label style={at={(axis description cs:0.5,1.0)},anchor=north west,xshift=0em},
        ]

            \draw [fill=red,draw=none,fill opacity=0.2] (axis cs:927,2e1) rectangle (axis cs:1077,7e5);
            \draw [fill=orange,draw=none,fill opacity=0.2] (axis cs:1195,2e1) rectangle (axis cs:1205,1e5);
            \draw [fill=orange,draw=none,fill opacity=0.2] (axis cs:1136,2e1) rectangle (axis cs:1146,1e5);
        
            \addplot [black,no marks,very thin] table [x=t,y=rx] {figures/prototype-experiments/traffic_runtime2500_n4100_f1025_Tcp300_13.53.192.225.txt};
            \label{leg:prototype-results-traffic-faults-rx-honest}

            \legend{};
            
        \end{axis}
    \end{tikzpicture}%
    \vspace{-0.3em}
    \caption[]{%
        Setting of Fig.~\ref{fig:prototype-results-ledgers-faults}(r):
        Leader timeouts in $\PIbft$ and $\PIacc$
        can delay new checkpoints (red). \Eg,
        after the end of a checkpoint interval
        ($t\approx870\,\mathrm{s}$),
        and subsequent
        $\PIacc$ leader timeout ($t\approx930\,\mathrm{s}$),
        honest nodes vote to
        reject the current checkpoint iteration,
        but the decision is delayed by another
        $\PIbft$ leader timeout.
        The next checkpoint iteration has an honest
        leader, but a decision is again delayed by
        a $\PIbft$ leader timeout,
        until a new checkpoint is finally reached
        ($t\approx1070\,\mathrm{s}$).
        Traffic at honest nodes (\textbf{right})
        lacks some of the small spikes (orange) 
        of
        traffic at adversarial nodes (\textbf{left}),
        since the adversary temporarily withholds
        some of its blocks
        from honest nodes due to selfish mining.%
    }
    \label{fig:prototype-results-traffic-faults}
\end{figure}
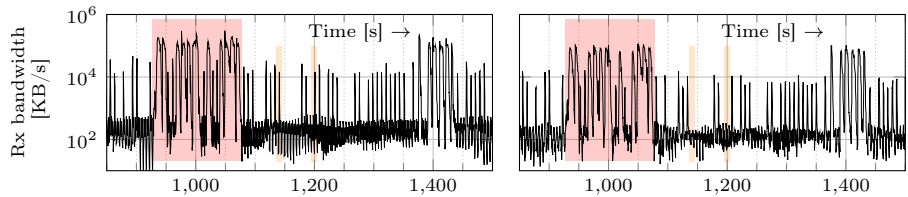

Network traffic
(Figs.~\ref{fig:prototype-results-traffic-faults}, \ref{fig:prototype-results-traffic-nofaults} \emph{for an exemplary AWS instance},
\ie, for $82$ nodes)
shows
frequent small spikes 
for
$\PIlc$ blocks and infrequent wide spikes 
for $\PIacc$ votes and $\PIbft$ blocks and votes.
Traffic increases slightly under attack
(per node:
avg. $78\,\mathrm{KB/s}$ vs. $56\,\mathrm{KB/s}$,
peak $1.5\,\mathrm{MB/s}$ vs. $1.34\,\mathrm{MB/s}$)
because inactive adversarial leaders
cause 
more
iterations 
in $\PIacc$ and $\PIbft$.
The bandwidth requirement 
does not limit participation
using
consumer-grade Internet access.
Note that our prototype does not employ bandwidth reduction
techniques that are orthogonal to the consensus problem,
such as aggregate and short signatures or spreading the vote
out over time.
Figure~\ref{fig:latency-bandwidth} corroborates that
even if voting 
was artificially rate-limited and thus spread out over time
(as is the case in Gasper),
bandwidth and latency comparable
to Gasper could be achieved.

\begin{figure}[t]
    \centering
    \begin{tikzpicture}
        \scriptsize
        \begin{axis}[
            mysimpleplot,
            xlabel={Avg. $\LOGacc{}{}$ latency [s] $\rightarrow$},
            ylabel={Avg. bandwidth\\requirement [vote/s]},
            xmax=2400,
            width=0.75\linewidth,
            height=0.45\linewidth,
            ymode=log,
            xmode=log,
            xtick={400,600,900,1350,2025},
            xticklabels={400,600,900,1350,2025},
            legend columns=1,
            legend style={
                at={(1,1)},
                xshift=0.3em,
                anchor=north west,
                draw=none,
                /tikz/every even column/.append style={
                    column sep=0.3em
                },
                cells={
                    align=left,
                },
                /tikz/every odd column/.style={yshift=0.5em},
            },
            x label style={at={(axis description cs:0.0,0.0)},anchor=south west,xshift=0.5em},
        ]

            \addplot [
                nodes near coords,
                point meta=explicit symbolic,
                visualization depends on={value \thisrow{style}\as\mystyle},
                every node near coord/.append style={\mystyle},
                myparula43,
                only marks,
                forget plot,
            ] table [
                x=LOGacc_latency,
                y=vote_bandwidth,
                meta=label,
            ] {
label LOGacc_latency vote_bandwidth style
$16$ 480.0 42.666666666666664 {below left}
$32$ 960.0 21.333333333333332 {below left}
$64$ 1920.0 10.666666666666666 {below left}
            };
            
            \addplot [
                myparula43,
                no marks,
                forget plot,
            ] table [
                x=LOGacc_latency,
                y=vote_bandwidth,
                meta=label,
            ] {
label LOGacc_latency vote_bandwidth
$16$ 480.0 42.666666666666664
$17$ 510.0 40.15686274509804
$18$ 540.0 37.925925925925924
$19$ 570.0 35.92982456140351
$20$ 600.0 34.13333333333333
$21$ 630.0 32.507936507936506
$22$ 660.0 31.03030303030303
$23$ 690.0 29.681159420289855
$24$ 720.0 28.444444444444443
$25$ 750.0 27.30666666666667
$26$ 780.0 26.25641025641026
$27$ 810.0 25.28395061728395
$28$ 840.0 24.38095238095238
$29$ 870.0 23.54022988505747
$30$ 900.0 22.755555555555556
$31$ 930.0 22.021505376344084
$32$ 960.0 21.333333333333332
$33$ 990.0 20.68686868686869
$34$ 1020.0 20.07843137254902
$35$ 1050.0 19.504761904761903
$36$ 1080.0 18.962962962962962
$37$ 1110.0 18.45045045045045
$38$ 1140.0 17.964912280701753
$39$ 1170.0 17.504273504273502
$40$ 1200.0 17.066666666666666
$41$ 1230.0 16.650406504065042
$42$ 1260.0 16.253968253968253
$43$ 1290.0 15.875968992248062
$44$ 1320.0 15.515151515151516
$45$ 1350.0 15.170370370370371
$46$ 1380.0 14.840579710144928
$47$ 1410.0 14.52482269503546
$48$ 1440.0 14.222222222222221
$49$ 1470.0 13.931972789115646
$50$ 1500.0 13.653333333333334
$51$ 1530.0 13.38562091503268
$52$ 1560.0 13.12820512820513
$53$ 1590.0 12.880503144654087
$54$ 1620.0 12.641975308641975
$55$ 1650.0 12.412121212121212
$56$ 1680.0 12.19047619047619
$57$ 1710.0 11.976608187134502
$58$ 1740.0 11.770114942528735
$59$ 1770.0 11.570621468926554
$60$ 1800.0 11.377777777777778
$61$ 1830.0 11.191256830601093
$62$ 1860.0 11.010752688172042
$63$ 1890.0 10.835978835978835
$64$ 1920.0 10.666666666666666
            };
            
            \addlegendimage{myparula43}
            \addlegendentry{Gasper\\$n=4096$, vary $C$}

            \addplot [
                nodes near coords,
                point meta=explicit symbolic,
                visualization depends on={value \thisrow{style}\as\mystyle},
                every node near coord/.append style={\mystyle},
                myparula42,
                only marks,
                forget plot,
            ] table [
                x=LOGacc_latency,
                y=vote_bandwidth,
                meta=label,
            ] {
label LOGacc_latency vote_bandwidth style
$16$ 480.0 85.33333333333333 {above right}
$32$ 960.0 42.666666666666664 {above right}
$64$ 1920.0 21.333333333333332 {above right}
            };
            
            \addplot [
                myparula42,
                no marks,
                forget plot,
            ] table [
                x=LOGacc_latency,
                y=vote_bandwidth,
                meta=label,
            ] {
label LOGacc_latency vote_bandwidth
$16$ 480.0 85.33333333333333
$17$ 510.0 80.31372549019608
$18$ 540.0 75.85185185185185
$19$ 570.0 71.85964912280701
$20$ 600.0 68.26666666666667
$21$ 630.0 65.01587301587301
$22$ 660.0 62.06060606060606
$23$ 690.0 59.36231884057971
$24$ 720.0 56.888888888888886
$25$ 750.0 54.61333333333334
$26$ 780.0 52.51282051282052
$27$ 810.0 50.5679012345679
$28$ 840.0 48.76190476190476
$29$ 870.0 47.08045977011494
$30$ 900.0 45.51111111111111
$31$ 930.0 44.04301075268817
$32$ 960.0 42.666666666666664
$33$ 990.0 41.37373737373738
$34$ 1020.0 40.15686274509804
$35$ 1050.0 39.009523809523806
$36$ 1080.0 37.925925925925924
$37$ 1110.0 36.9009009009009
$38$ 1140.0 35.92982456140351
$39$ 1170.0 35.008547008547005
$40$ 1200.0 34.13333333333333
$41$ 1230.0 33.300813008130085
$42$ 1260.0 32.507936507936506
$43$ 1290.0 31.751937984496124
$44$ 1320.0 31.03030303030303
$45$ 1350.0 30.340740740740742
$46$ 1380.0 29.681159420289855
$47$ 1410.0 29.04964539007092
$48$ 1440.0 28.444444444444443
$49$ 1470.0 27.86394557823129
$50$ 1500.0 27.30666666666667
$51$ 1530.0 26.77124183006536
$52$ 1560.0 26.25641025641026
$53$ 1590.0 25.761006289308174
$54$ 1620.0 25.28395061728395
$55$ 1650.0 24.824242424242424
$56$ 1680.0 24.38095238095238
$57$ 1710.0 23.953216374269005
$58$ 1740.0 23.54022988505747
$59$ 1770.0 23.141242937853107
$60$ 1800.0 22.755555555555556
$61$ 1830.0 22.382513661202186
$62$ 1860.0 22.021505376344084
$63$ 1890.0 21.67195767195767
$64$ 1920.0 21.333333333333332
            };
            
            \addlegendimage{myparula42}
            \addlegendentry{Gasper\\$n=8192$, vary $C$}

            \addplot [
                nodes near coords,
                point meta=explicit symbolic,
                visualization depends on={value \thisrow{style}\as\mystyle},
                every node near coord/.append style={\mystyle},
                myparula53,
                only marks,
                forget plot,
            ] table [
                x=LOGacc_latency,
                y=vote_bandwidth,
                meta=label,
            ] {
label LOGacc_latency vote_bandwidth style
$10\,\mathrm{min}$ 372.0 34.13333333333333 {above}
$20\,\mathrm{min}$ 672.0 17.066666666666666 {below left}
$30\,\mathrm{min}$ 972.0 11.377777777777778 {below left}
$40\,\mathrm{min}$ 1272.0 8.533333333333333 {below left}
$50\,\mathrm{min}$ 1572.0 6.826666666666667 {below left}
$60\,\mathrm{min}$ 1872.0 5.688888888888889 {below left}
            };
            
            \addplot [
                myparula53,
                no marks,
                forget plot,
            ] table [
                x=LOGacc_latency,
                y=vote_bandwidth,
                meta=label,
            ] {
label LOGacc_latency vote_bandwidth
$10\,\mathrm{min}$ 372.0 34.13333333333333
$11\,\mathrm{min}$ 402.0 31.03030303030303
$12\,\mathrm{min}$ 432.0 28.444444444444443
$13\,\mathrm{min}$ 462.0 26.256410256410255
$14\,\mathrm{min}$ 492.0 24.38095238095238
$15\,\mathrm{min}$ 522.0 22.755555555555556
$16\,\mathrm{min}$ 552.0 21.333333333333332
$17\,\mathrm{min}$ 582.0 20.07843137254902
$18\,\mathrm{min}$ 612.0 18.962962962962962
$19\,\mathrm{min}$ 642.0 17.964912280701753
$20\,\mathrm{min}$ 672.0 17.066666666666666
$21\,\mathrm{min}$ 702.0 16.253968253968253
$22\,\mathrm{min}$ 732.0 15.515151515151516
$23\,\mathrm{min}$ 762.0 14.840579710144928
$24\,\mathrm{min}$ 792.0 14.222222222222221
$25\,\mathrm{min}$ 822.0 13.653333333333334
$26\,\mathrm{min}$ 852.0 13.128205128205128
$27\,\mathrm{min}$ 882.0 12.641975308641975
$28\,\mathrm{min}$ 912.0 12.19047619047619
$29\,\mathrm{min}$ 942.0 11.770114942528735
$30\,\mathrm{min}$ 972.0 11.377777777777778
$31\,\mathrm{min}$ 1002.0 11.010752688172044
$32\,\mathrm{min}$ 1032.0 10.666666666666666
$33\,\mathrm{min}$ 1062.0 10.343434343434344
$34\,\mathrm{min}$ 1092.0 10.03921568627451
$35\,\mathrm{min}$ 1122.0 9.752380952380953
$36\,\mathrm{min}$ 1152.0 9.481481481481481
$37\,\mathrm{min}$ 1182.0 9.225225225225225
$38\,\mathrm{min}$ 1212.0 8.982456140350877
$39\,\mathrm{min}$ 1242.0 8.752136752136753
$40\,\mathrm{min}$ 1272.0 8.533333333333333
$41\,\mathrm{min}$ 1302.0 8.325203252032521
$42\,\mathrm{min}$ 1332.0 8.126984126984127
$43\,\mathrm{min}$ 1362.0 7.937984496124031
$44\,\mathrm{min}$ 1392.0 7.757575757575758
$45\,\mathrm{min}$ 1422.0 7.5851851851851855
$46\,\mathrm{min}$ 1452.0 7.420289855072464
$47\,\mathrm{min}$ 1482.0 7.26241134751773
$48\,\mathrm{min}$ 1512.0 7.111111111111111
$49\,\mathrm{min}$ 1542.0 6.965986394557823
$50\,\mathrm{min}$ 1572.0 6.826666666666667
$51\,\mathrm{min}$ 1602.0 6.69281045751634
$52\,\mathrm{min}$ 1632.0 6.564102564102564
$53\,\mathrm{min}$ 1662.0 6.440251572327044
$54\,\mathrm{min}$ 1692.0 6.320987654320987
$55\,\mathrm{min}$ 1722.0 6.206060606060606
$56\,\mathrm{min}$ 1752.0 6.095238095238095
$57\,\mathrm{min}$ 1782.0 5.988304093567251
$58\,\mathrm{min}$ 1812.0 5.885057471264368
$59\,\mathrm{min}$ 1842.0 5.785310734463277
$60\,\mathrm{min}$ 1872.0 5.688888888888889
            };
            
            \addlegendimage{myparula53}
            \addlegendentry{\Ourprotocol\\$n=4096$, vary $\Tcheckpoint$}

            \addplot [
                nodes near coords,
                point meta=explicit symbolic,
                visualization depends on={value \thisrow{style}\as\mystyle},
                every node near coord/.append style=\mystyle,
                myparula52,
                only marks,
                forget plot,
            ] table [
                x=LOGacc_latency,
                y=vote_bandwidth,
                meta=label,
            ] {
label LOGacc_latency vote_bandwidth style
$10\,\mathrm{min}$ 372.0 68.26666666666667 {above right}
$20\,\mathrm{min}$ 672.0 34.13333333333333 {above right}
$30\,\mathrm{min}$ 972.0 22.755555555555556 {above right}
$40\,\mathrm{min}$ 1272.0 17.066666666666666 {above right}
$50\,\mathrm{min}$ 1572.0 13.653333333333334 {above right}
$60\,\mathrm{min}$ 1872.0 11.377777777777778 {above right}
            };
            
            \addplot [
                myparula52,
                no marks,
                forget plot,
            ] table [
                x=LOGacc_latency,
                y=vote_bandwidth,
                meta=label,
            ] {
label LOGacc_latency vote_bandwidth
$10\,\mathrm{min}$ 372.0 68.26666666666667
$11\,\mathrm{min}$ 402.0 62.06060606060606
$12\,\mathrm{min}$ 432.0 56.888888888888886
$13\,\mathrm{min}$ 462.0 52.51282051282051
$14\,\mathrm{min}$ 492.0 48.76190476190476
$15\,\mathrm{min}$ 522.0 45.51111111111111
$16\,\mathrm{min}$ 552.0 42.666666666666664
$17\,\mathrm{min}$ 582.0 40.15686274509804
$18\,\mathrm{min}$ 612.0 37.925925925925924
$19\,\mathrm{min}$ 642.0 35.92982456140351
$20\,\mathrm{min}$ 672.0 34.13333333333333
$21\,\mathrm{min}$ 702.0 32.507936507936506
$22\,\mathrm{min}$ 732.0 31.03030303030303
$23\,\mathrm{min}$ 762.0 29.681159420289855
$24\,\mathrm{min}$ 792.0 28.444444444444443
$25\,\mathrm{min}$ 822.0 27.30666666666667
$26\,\mathrm{min}$ 852.0 26.256410256410255
$27\,\mathrm{min}$ 882.0 25.28395061728395
$28\,\mathrm{min}$ 912.0 24.38095238095238
$29\,\mathrm{min}$ 942.0 23.54022988505747
$30\,\mathrm{min}$ 972.0 22.755555555555556
$31\,\mathrm{min}$ 1002.0 22.021505376344088
$32\,\mathrm{min}$ 1032.0 21.333333333333332
$33\,\mathrm{min}$ 1062.0 20.68686868686869
$34\,\mathrm{min}$ 1092.0 20.07843137254902
$35\,\mathrm{min}$ 1122.0 19.504761904761907
$36\,\mathrm{min}$ 1152.0 18.962962962962962
$37\,\mathrm{min}$ 1182.0 18.45045045045045
$38\,\mathrm{min}$ 1212.0 17.964912280701753
$39\,\mathrm{min}$ 1242.0 17.504273504273506
$40\,\mathrm{min}$ 1272.0 17.066666666666666
$41\,\mathrm{min}$ 1302.0 16.650406504065042
$42\,\mathrm{min}$ 1332.0 16.253968253968253
$43\,\mathrm{min}$ 1362.0 15.875968992248062
$44\,\mathrm{min}$ 1392.0 15.515151515151516
$45\,\mathrm{min}$ 1422.0 15.170370370370371
$46\,\mathrm{min}$ 1452.0 14.840579710144928
$47\,\mathrm{min}$ 1482.0 14.52482269503546
$48\,\mathrm{min}$ 1512.0 14.222222222222221
$49\,\mathrm{min}$ 1542.0 13.931972789115646
$50\,\mathrm{min}$ 1572.0 13.653333333333334
$51\,\mathrm{min}$ 1602.0 13.38562091503268
$52\,\mathrm{min}$ 1632.0 13.128205128205128
$53\,\mathrm{min}$ 1662.0 12.880503144654089
$54\,\mathrm{min}$ 1692.0 12.641975308641975
$55\,\mathrm{min}$ 1722.0 12.412121212121212
$56\,\mathrm{min}$ 1752.0 12.19047619047619
$57\,\mathrm{min}$ 1782.0 11.976608187134502
$58\,\mathrm{min}$ 1812.0 11.770114942528735
$59\,\mathrm{min}$ 1842.0 11.570621468926554
$60\,\mathrm{min}$ 1872.0 11.377777777777778
            };
            
            \addlegendimage{myparula52}
            \addlegendentry{\Ourprotocol\\$n=8192$, vary $\Tcheckpoint$}

        \end{axis}
    \end{tikzpicture}%
    \vspace{-0.3em}
    \caption[]{%
        For fixed $n$,
        the average latency of $\LOGacc{}{}$
        for Gasper and \ourprotocol (here for $\kcp=6$)
        increases with the number $C$ of slots per epoch
        and with $\Tcheckpoint$,
        respectively,
        while the bandwidth required for votes
        reduces proportionally.
        \Ourprotocol offers a better
        tradeoff
        and can tolerate twice the $n$
        at comparable latency and bandwidth
        (\ourprotocol for $n=8192, \Tcheckpoint=30\,\mathrm{min}$ vs. Gasper for $n=4096, C=32$).%
    }
    \label{fig:latency-bandwidth}
\end{figure}

Figure~\ref{fig:latency-bandwidth} compares
bandwidth
and latency of $\LOGacc{}{}$ for varying parameters
and $\beta = 0, \Delta = 0$.
Gasper transmits
$2 \cdot \frac{n}{C}$ votes per $12\,\mathrm{s}$,
with $C$ the number of slots per epoch,
\ourprotocol transmits $5 \cdot n$ votes per $\Tcheckpoint$ time.
A transaction takes on average $\frac{1}{2} + 2$ epochs to enter
into 
$\LOGacc{}{}$
for Gasper,
and $\kcp \cdot 12\,\mathrm{s} + \frac{1}{2}\cdot\Tcheckpoint$
time to enter 
$\LOGacc{}{}$
for \ourprotocol.
\Ourprotocol offers slightly improved latency at comparable
bandwidth, or comparable bandwidth and latency but for
a larger number of nodes.

\section*{Acknowledgment}
JN, ENT, and DT
are supported by
the Reed-Hodgson Stanford Graduate Fellowship,
the Stanford Center for Blockchain Research,
and 
the Center for Science of Information (CSoI), an NSF Science and Technology Center under grant agreement CCF-0939370, respectively.
\bibliographystyle{splncs04}
\bibliography{references}

\appendix

\section{Proof of \AAD}
\label{sec:appendix-aad-proof}

A formal proof of Theorem~\ref{thm:aa-dilemma} building on the observations discussed in Section~\ref{sec:availability-accountability-dilemma} is as follows:
\begin{proof}
For the sake of contradiction, suppose there exists an SMR protocol $\PI$ that provides $\betaL$-liveness and $\betaA$-accountable-safety for some $\betaL,\betaA>0$ under $\AdvEnvDA$.
Then, there exists an adjudication function $\adj$, which given two sets of \clues attesting to conflicting ledgers, outputs a non-empty set of adversarial nodes.

Suppose there are $n$ nodes in $\Env$.
Without loss of generality, we may assume that $n$ is even; otherwise, 
$\Env$ puts one
node to sleep throughout the execution.
Let 
$P$ and $Q$ partition the $n$ nodes into two disjoint equal groups
with $|P|=|Q|=n/2$.
We denote by $[\tx]$ a ledger consisting of a single transaction $\tx$ at its first index.

Next consider the following worlds:

\textbf{World 1:}
Nodes in $P$ are honest and awake throughout the execution.
$\Env$ inputs $\tx_1$ to them.
Nodes in $Q$ are asleep.
Since $\PI$ satisfies liveness for some $\betaL > 0$ under $\AdvEnvDA$, nodes in $P$ eventually generate a set of \clues $W_1$ such that $\Cf(W_1)=[\tx_1]$.

\textbf{World 2:}
Nodes in $Q$ are honest and awake throughout the execution.
$\Env$ inputs $\tx_2$ to them.
Nodes in $P$ are asleep.
Since $\PI$ satisfies liveness for some $\betaL > 0$ under $\AdvEnvDA$, nodes in $Q$ eventually generate a set of \clues $W_2$ such that $\Cf(W_2)=[\tx_2]$.

\textbf{World 3:} $\Env$ wakes up all $n$ nodes, and inputs $\tx_1$ to the nodes in $P$ and $\tx_2$ to the nodes in $Q$.
Nodes in $P$ are honest.
Nodes in $Q$ are adversarial and do not communicate with the nodes in $P$.
All nodes stay awake throughout the execution.
Since the worlds 1 and 3 are indistinguishable for the nodes in $P$, they eventually generate a set of \clues $W_1$ such that $\Cf(W_1)=[\tx_1]$.
Nodes in $Q$ simulate the execution in world 2 without any communication with the nodes in $P$.
Hence, they eventually generate a set of \clues $W_2$ such that $\Cf(W_2)=[\tx_2]$.
Thus, there is a safety violation.
So,
$\adj$ takes $W_1$ and $W_2$, and outputs a non-empty set $S_3 \subseteq Q$ of adversarial nodes.

\textbf{World 4:} $\Env$ wakes up all $n$ nodes, and inputs $\tx_1$ to the nodes in $P$ and $\tx_2$ to the nodes in $Q$.
Nodes in $Q$ are honest.
Nodes in $P$ are adversarial and do not communicate with the nodes in $Q$.
All nodes stay awake throughout the execution.
Since the worlds 2 and 4 are indistinguishable for the nodes in $Q$, they eventually generate a set of \clues $W_2$ such that $\Cf(W_2)=[\tx_2]$.
Nodes in $P$ simulate the execution in world 1 without any communication with the nodes in $Q$.
Hence, they eventually generate a set of \clues $W_1$ such that $\Cf(W_1)=[\tx_1]$.
Thus, there is a safety violation.
So,
$\adj$ takes $W_1$ and $W_2$, and outputs a non-empty set $S_4 \subseteq P$ of adversarial nodes.

Note however
that worlds 3 and 4 are indistinguishable from the perspective of the adjudication function $\adj$.
Thus, it is not possible that $\adj$ reliably outputs
a non-empty set which in the case of world 3
contains only elements of $Q$ and in the case
of world 4 contains only elements of $P$,
as would be required by Definition~\ref{def:adjudication-protocol}.
\end{proof}

\section{Security Proof for Accountability Gadgets}

\label{sec:appendix-security-proofs}

\subsection{Theorem Statement and Notation}
\label{sec:appendix-notation}

In this section, we consider an accountability gadget $\PIacc$ instantiated with a BFT protocol $\PIbft$ that provides $(n-2f)$-accountable-safety at all times, and $f$-liveness after  $\max(\GST,\GAT)$ under
$\AdvEnvPG$.
To match the accountable safety resilience of $\PIbft$ on $\PIacc$, we tune the thresholds for the number of accept and reject votes required to output a new checkpoint as $(n-f)$ and $f+1$ respectively on lines~\ref{line:threshold1} and~\ref{line:threshold2} of Algorithm~\ref{algo:pseudocode-checkpoint-extractor}.

Recall that $\PIacc$ is used on top of a Nakamoto-style permissioned longest chain (LC) protocol $\PIlc$.
For concreteness and notational purposes, we assume that $\PIacc$ is the Sleepy consensus protocol \cite{sleepy} although  we could have used any other permissioned LC protocol in its place.

Given the accountability gadget $\PIacc$ and the LC protocol $\PIlc$, goal of this section is to prove that the ledgers $\LOGacc{}{}$ and $\LOGda{}{}$ outputted by $\PIacc$ and $\PIlc$ satisfy Theorem~\ref{thm:main-security} repeated below:

Given any security parameter $\sigma$ and $f \leq \lceil n/2 \rceil$,
\begin{enumerate}
    \item (\textbf{P1:Accountability}) Under $\AdvEnvPG$, the accountable ledger $\LOGacc{}{}$ provides $n-2f$-accountable safety at all times, and there exists a constant $\mathbf{C}$ such that $\LOGacc{}{}$ provides $f$-liveness (with confirmation time polynomial in $\sigma$) after \linebreak $\mathbf{C}\max(\GST,\GAT)$ except with probability $\negl(\sigma)$.
    \item (\textbf{P2:Dynamic Availability}) Under $\AdvEnvDA$, the available ledger $\LOGda{}{}$ is guaranteed to be safe and live at all times, provided that $\beta<1/2$.
    \item (\textbf{Prefix}) $\LOGacc{}{}$ is always a prefix of $\LOGda{}{}$.
\end{enumerate}

Before proceeding with the proofs, we formalize the concept of \emph{security after a certain time} (We write $\LOG{}{} \preceq \LOGBLANKFIX'$ if $\LOG{}{}$ is a prefix of $\LOGBLANKFIX'$.):

\begin{definition}
\label{def:security-after}
Let $\Tconfirm$ be a polynomial function of the security parameter $\sigma$. 
We say that a ledger $\LOG{}{}$ is \emph{secure after time} $T$ and has transaction confirmation time $\Tconfirm$
if $\LOG{}{}$ satisfies:
\begin{itemize}
    \item \textbf{Safety:} For any two times $t \geq t' \geq T$, and any two honest nodes $i$ and $j$ awake at times $t$ and $t'$ respectively, either $\LOG{i}{t} \preceq \LOG{j}{t'}$ or $\LOG{j}{t'} \preceq \LOG{i}{t}$.
    \item \textbf{Liveness:} If a transaction is received by an awake honest node at some time $t \geq T$, then, for any time $t' \geq t+\Tconfirm$ and honest node $j$ that is awake at time $t'$, the transaction will be included in $\LOG{j}{t'}$.
\end{itemize}
\end{definition}

Definition~\ref{def:security-after} formalizes the meaning of `safety, liveness and security \emph{after a certain time $T$}'.
In general, there might be two different times after which a protocol is safe or live.
A protocol that is safe (live) at all times (i.e, after $T=0$) is simply called \emph{safe} (\emph{live}) without further qualification.

\subsection{Accountable Safety Resilience}
\label{sec:appendix-accountable-safety}

We first show that $\LOGacc{}{}$ provides $n-2f$-accountable safety under $\AdvEnvPG$.

\begin{proposition}
\label{thm:checkpoint-safety}
Suppose the number of adversarial nodes is less than $n-2f$. Then, if a block $b$ is checkpointed for iteration $c$ in the view of an honest node $i$ at slot $t$, for any honest node $j$ and slot $s$, either $b$ is checkpointed for iteration $c$ at slot $s$ or no block has been checkpointed for iteration $c$ yet.
\end{proposition}

Proposition~\ref{thm:checkpoint-safety} follows from the safety of $\LOGbft{}{}$ when the number of adversarial nodes is less than $n-2f$.

\begin{theorem}[Accountable Safety of $\LOGacc{}{}$]
\label{thm:accountable-safety-theorem}
$\LOGacc{}{}$ provides $n-2f$-accountable-safety.
\end{theorem}

\begin{proof}
To show that $\LOGacc{}{}$ provides $n-2f$-accountable-safety, we construct an adjudication protocol $\adj$, which in the case of a safety violation on $\LOGacc{}{}$, outputs at least $n-2f$ nodes as adversarial and never outputs an honest node.
For this purpose, suppose there is a safety violation on $\LOGacc{}{}$.
Then, there exist honest nodes $i$ and $j$, iterations $c$ and $c'$ (without loss of generality $c' \leq c$) and slots $s$ and $t$ such that (i) a block $b_1 \neq \bot$ is checkpointed for iteration $c'$ in the view of node $i$ at slot $s$, (ii) a block $b_2 \neq \bot$ is checkpointed for iteration $c$ in the view of node $j$ at slot $t$, (iii) $b_1,b_2$ conflict with each other.
Then, within $\LOGbft{s}{i}$, there are accept votes $\CpReq{accept}{c',b_1}{k}$ from at least $n-f$ nodes for the proposal $b_1$ and iteration $c'$.
Similarly, within $\LOGbft{t}{j}$, there are accept votes $\CpReq{accept}{c,b_2}{k}$ from at least $n-f$ nodes for the proposal $b_2$ and iteration $c$.
Thus, more than $2(n-f)-n=n-2f$ nodes voted both $\CpReq{accept}{c',b_1}{k}$ and $\CpReq{accept}{c,b_1}{k}$.
Let $S$ denote the set of these nodes.

Next, consider the following two cases:

(i) There is a safety violation on $\LOGbft{}{}$.
(Recall that $\LOGbft{}{}$ provides $n-2f$-accountable safety.)
In this case, since the adjudication protocol for $\LOGbft{}{}$ identifies at least $n-2f$ nodes as adversarial, $\adj$ simply returns the output of the adjudication protocol for $\LOGbft{}{}$. 

(ii) Suppose there is no safety violation on $\LOGbft{}{}$ and $c'<c$. 
Then, via Proposition~\ref{thm:checkpoint-safety}, every node $k$ in $S$ have either seen $b_1$ become checkpointed for iteration $c'$ before voting accept $\CpReq{accept}{c,b_2}{k}$ or voted for iteration $c$ before seeing any checkpoint for iteration $c'$.
However, an honest node votes accept for the proposal $b_2$ of iteration $c$ only if it has already seen a block checkpointed for all past iterations including $c'$, and if $b_2$ is consistent with all of the checkpoints from the past iterations, including $b_1$.
(This is because an honest node votes accept for a proposal only if it is part of the node's checkpoint-respecting LC.)
Then, no honest node could have voted both $\CpReq{accept}{c',b_1}{k}$ and $\CpReq{accept}{c,b_2}{k}$, implying that all of the nodes in $S$ have violated the protocol.
If $c'=c$, then all of the nodes in $S$ voted accept twice for two different proposals for iteration $c$, which is again a protocol violation.

Finally, when there is no safety violation on $\LOGbft{}{}$, $\LOGbft{i}{s} \preceq \LOGbft{j}{t}$ or $\LOGbft{j}{t} \preceq \LOGbft{i}{s}$.
In either case, all of the accept votes $\CpReq{accept}{c',b_1}{k}$ and $\CpReq{accept}{c,b_2}{k}$ are within the longer of $\LOGbft{i}{s}$ and $\LOGbft{j}{t}$, which can be used to prove that the nodes in $S$ violated the protocol.
Hence, in this case, $\adj$ returns $S$, which contains at least $n-2f$ nodes as the set of nodes that have irrefutably violated the protocol. 
\end{proof}

\subsection{Liveness Resilience}
\label{sec:appendix-liveness}

We next focus on the liveness of $\LOGacc{}{}$.

\begin{proposition}
\label{thm:bft-live}
$\PIbft$ satisfies $f$-liveness after $\max(\GST,\GAT)$ with transaction confirmation time $\Tconfirm$.
\end{proposition}

\begin{proposition}
\label{prop:entering-iter-after-gst}
Suppose a block from iteration $c$ was checkpointed in the view of an honest node at time $t$.
Then, every honest node enters iteration $c+1$ by time $\max(\GST,\GAT,t)+\Delta$.
\end{proposition}
\begin{proof}
Suppose a block from iteration $c$ was checkpointed in the view of an honest node $i$ at time $t$.
Then, there are at least $n-f$ $\mathsf{accept}$ votes for the block from iteration $c$ on $\LOGbft{i}{t}$.
Note that all checkpointing votes and BFT protocol messages observed by node $i$ by time $t$ are delivered to all other honest nodes by time $\max(\GST,\GAT,t)+\Delta$.
Hence, via Theorem~\ref{thm:accountable-safety-theorem}, for any honest node $j$, $\LOGbft{t}{i} \preceq \LOGbft{\max(\GST,\GAT,t)+\Delta}{j}$.
Thus, for any honest node $j$, there are at least $n-f$ $\mathsf{accept}$ votes on $\LOGbft{\max(\GST,\GAT,t)+\Delta}{j}$ for the same block from iteration $c$.
This implies that every honest node enters iteration $c+1$ by time $\max(\GST,\GAT,t)+\Delta$.
\end{proof}

\begin{theorem}[Liveness of $\LOGacc{}{}$]
\label{thm:liveness-lemma}
Suppose $\PIlc$ is secure (safe and live) after some slot $T \geq \max(\GST,\GAT)+\Delta+\Tcheckpoint$.
Then, $\LOGacc{}{}$ satisfies $f$-liveness after slot $T$
with transaction confirmation time $\sigma$ except with probability $e^{-\Omega(\sigma)}$.
\end{theorem}

\begin{proof}
If there are $f$ or more adversarial nodes, we know via Proposition~\ref{thm:bft-live} that $\LOGbft{}{}$, and by implication $\LOGacc{}{}$ will not be live. 
Thus, to show $f$-liveness of $\LOGacc{}{}$, we assume that there are less than $f$ adversarial nodes and prove that $\LOGacc{}{}$ satisfies liveness.
In this case, again via Proposition~\ref{thm:bft-live}, we know that $\LOGbft{}{}$ satisfies liveness with transaction confirmation time $\Tconfirm$ after $\max(\GST,\GAT)$, a property which we will use subsequently.

Let $c'$ be the largest iteration such that a block was checkpointed in the view of some honest node before $\max(\GAT,\GST)$.
(Let $c'=0$ if there does not exist such an iteration.)
By Proposition~\ref{prop:entering-iter-after-gst}, all honest nodes would have entered iteration $c'+1$ by slot $\max(\GAT,\GST)+\Delta$.
Then, all honest nodes observe a block proposed for iteration $c'+1$ by slot $\max(\GAT,\GST)+\Delta+\Tcheckpoint$.
Thus, entrance times of the honest nodes to subsequent iterations have become synchronized by slot  $\max(\GAT,\GST)+\Delta+\Tcheckpoint$:
If an honest node enters an iteration $c>c'$ at slot $t \geq \max(\GAT,\GST)+\Delta+\Tcheckpoint$, every honest node enters iteration $c'$ by slot $t+\Delta$.

Suppose iteration $c>c'$ has an honest leader $\ld{c}$, which sends a proposal $\bprop{c}$ at slot $t > T + \Tcheckpoint$.
Note that $\bprop{c}$ is received by every honest node by slot $t+\Delta$.
Since the entrance times of nodes are synchronized by $T \leq \max(\GST,\GAT)+\Delta+\Tcheckpoint$, every honest node votes by slot $t+\Delta$.
Now, as $\PIlc$ is secure after slot $T$, $\bprop{c}$ is on all of the checkpoint-respecting LCs held by the honest nodes.
Moreover, as we will argue in the paragraph below, $\bprop{c}$ extends all of the checkpoints seen by the honest nodes by slot $t+\Delta$.
(Honest nodes see the same checkpoints from iterations preceding $c$ due to synchrony.)
Consequently, every honest node votes $\CpReq{accept}{c,\bprop{c}}{k}$ for $\bprop{c}$ by slot $t+\Delta$, all of which appear within $\LOGbft{}{}$ in the view of every honest node by slot $t+\Delta+\Tconfirm$.
Thus, $\bprop{c}$ becomes checkpointed in the view of every honest node by slot $t+\Delta+\Tconfirm$. 
(Here, we assume that $\Ttimeout$ was chosen large enough for $\Ttimeout>\Delta+\Tconfirm$ to hold.)

Note that $\ld{c}$ waits for $\Tcheckpoint$ slots before broadcasting $\bprop{c}$ after observing the last checkpoint block before iteration $c$.
Since $t-\Tcheckpoint>T$, during the period $[t-\Tcheckpoint,t]$, $\PIlc$ satisfies the chain growth and quality properties (see Appendix~\ref{sec:appendix-security-proof-for-the-longest-chain-protocol}). 
Thus, for a large enough $\Tcheckpoint$, the checkpoint-respecting LC of $\ld{c}$ at time $t$ contains at least one honest block between $\bprop{c}$ and the last checkpointed block on it from before iteration $c$.
(As a corollary, $\ld{c}$ extends all of the previous checkpoints seen by itself and all other honest nodes.)
This implies that $\bprop{c}$ contains at least one fresh honest block that enters $\LOGacc{}{}$ by slot $t+\Delta+\Tconfirm$.

Next, we show that an adversarial leader cannot make an iteration last longer than $\Delta+\Ttimeout+\Tconfirm$ for any honest node.
Indeed, if an honest node $i$ enters an iteration $c$ at slot $t$, by slot $t+\Delta+\Ttimeout$, either it sees a block become checkpointed for iteration $c$, or every honest node votes reject.
In the first case, every honest node sees a block checkpointed for iteration $c$ by slot at most $t+2\Delta+\Ttimeout$.
In the second case, reject votes $\CpReq{reject}{c}{k}$ from at least $n-f > f$ of the nodes appear in $\LOGbft{}{}$ in the view of every honest node by slot at most $t+\Delta+\Ttimeout+\Tconfirm$.
Hence, a new checkpoint, potentially $\bot$, is outputted in the view of every honest node by slot $t+\Delta+\Ttimeout+\Tconfirm$.

Finally, we observe that except with probability $(f/n)^m$, there exist an iteration with an honest leader within $m$ consecutive iterations.
Since an iteration lasts at most $\max(\Delta+\Ttimeout+\Tconfirm,\Delta+\Tconfirm+\Tcheckpoint) \leq \Delta+\Ttimeout+\Tconfirm+\Tcheckpoint$ slots and a new checkpoint containing a fresh honest block in its prefix appears when an iteration has an honest leader, any transaction received by an awake honest node at slot $t$ appears within $\LOGacc{}{}$ in the view of every honest node by slot at most $\max(t,T)+m(\Delta+\Ttimeout+\Tconfirm+\Tcheckpoint)$ except with probability $(f/n)^m$.
Hence, via a union bound over the total number of iterations (which is a polynomial in $\sigma$), we observe that if $\PIlc$ is secure after some slot $T$, then $\LOGacc{}{}$ satisfies liveness after $T$ with a transaction confirmation time polynomial in $\sigma$ except with probability $e^{-\Omega(\sigma)}$.
\end{proof}

Observe that Theorem~\ref{thm:liveness-lemma} requires $\PIlc$ to eventually regain its security under $\AdvEnvPG$ when there are less than $f$ adversarial nodes.
Although it is not possible to guarantee any security property for $\PIlc$ before $\GST$, the following theorem states that $\PIlc$ recovers its security after $\max(\GST,\GAT)$.
Note that $\PIlc$ is initialized with a parameter $p$ which denotes the probability that a given node is elected as a block producer in a given slot.

\begin{theorem}[Security of $\PIlc$]
\label{thm:security-of-PIlc}
If $p<(n-2f)/(2\Delta n (n-f))$ and there are $f$ (or less) adversarial nodes, for each sufficiently large $\Tcheckpoint$, there exists a constant $\mathbf{C} > 0$ such that for any $\GST$ and $\GAT$ specified by $\AdvEnvPG$, $\PIlc(p)$ is secure after $\mathbf{C}(\max(\GST, \GAT) + \sigma)$, with transaction confirmation time $\sigma$, except with probability $e^{-\Omega(\sqrt{\sigma})}$.
\end{theorem}

Proof of the theorem is given in Section~\ref{sec:appendix-security-proof-for-the-longest-chain-protocol} and relies on a combination of the method outlined in \cite[Appendix C]{sankagiri_clc} with the concept of strong pivots from \cite{sleepy}.

Finally, since $f < n/2$, we can always find a $p$ such that $p<(n-2f)/(2\Delta n (n-f))$.
Then, given Theorems~\ref{thm:security-of-PIlc} and~\ref{thm:liveness-lemma} and a sufficiently small $p$, we can assert that $\LOGacc{}{}$ satisfies $f$-liveness with a transaction confirmation time polynomial in $\sigma$ after time $\mathbf{C}(\max(\GST,\GAT)+\sigma)$ except with probability $e^{-\Omega(\sqrt{\sigma})}$.

\subsection{Recency and Gap Properties}
\label{sec:appendix-gap-and-recency}

Proof of Theorem~\ref{thm:security-of-PIlc} requires the accountability gadget $\PIacc$ to satisfy two main properties first introduced in \cite{sankagiri_clc}: \emph{recency} and \emph{gap} properties.

Gap property states that blocks are checkpointed sufficiently apart in time, controlled by the parameter $\Tcheckpoint$:
\begin{proposition}[Gap Property]
\label{thm:checkpoint-gap}
Given any time interval $[t_1,t_2]$, no more than $(1+(t_2-t_1))/\Tcheckpoint$ blocks can be checkpointed in the interval. 
\end{proposition}
Proof of Proposition~\ref{thm:checkpoint-gap} follows from the fact that upon observing a new checkpoint that is not $\bot$ for an iteration, honest nodes wait for $\Tcheckpoint$ slots before voting for the proposal of the next iteration.

Following the notation in \cite{sankagiri_clc}, we say that a block checkpointed for iteration $c$ at slot $t>\max(\GST,\GAT)$ in the view of an honest node $i$ is $\Trecent$-recent if it has been part of the checkpoint-respecting LC of some honest node $j$ at some slot within $[t-\Trecent,t]$.
Then, we can express the recency property as follows:
\begin{lemma}[Recency Property]
\label{thm:checkpoint-recency}
Every checkpointed block proposed after \linebreak $\max(\GST,\GAT)$ is $\Trecent$-recent for $\Trecent = \Delta+\Ttimeout+\Tconfirm$.
\end{lemma}

\begin{proof}
We have seen in the proof of Theorem~\ref{thm:liveness-lemma} that if a block $b$ proposed after $\max(\GST,\GAT)$ is checkpointed in the view of an honest node at slot $t$, it should have been proposed after slot $t-(\Delta+\Ttimeout+\Tconfirm)$.
Thus, at least $n-f$ nodes must have voted $\CpReq{accept}{c,b}{k}$ by time $t$.
Let $j$ denote an honest nodes which voted $\CpReq{accept}{c,b}{j}$.
Note that $j$ would vote only after it sees the proposal for iteration $c$, i.e after slot $t-\Trecent = t-(\Delta+\Ttimeout+\Tconfirm)$.
Hence, $b$ should have been on the checkpoint-respecting LC of node $j$ at some slot within $[t-\Trecent,t]$.
This concludes the proof that every checkpoint block proposed after $\max(\GST,\GAT)$ is $\Trecent$-recent.
\end{proof}

\section{Security Proof for Checkpoint-Respecting LC}
\label{sec:appendix-security-proof-for-the-longest-chain-protocol}

In this section, we prove Theorem~\ref{thm:security-of-PIlc}, which states that the security of $\PIlc$ is restored after $\max(\GST,\GAT)$ under $\AdvEnvPG$ provided that the election probability $p$ of each node is sufficiently small.
Via \cite[Appendix C]{ebbandflow}, we know that the Sleepy consensus protocol \cite{sleepy} regains its safety and liveness within $O(\max(\GST, \GAT))$ slots under $\AdvEnvPG$.
On a similar note, \cite{redux} has shown the same self-healing property for Nakamoto's PoW LC protocol and other LC based PoS protocols.
However, all of these works analyze the LC protocols in their original form without considering checkpoints in the chain selection rule.
In this context, \cite{sankagiri_clc} is the first work to show the recovery of security for checkpoint-respecting LC protocols.
As argued in \cite{sankagiri_clc}, length of a checkpoint-respecting LC held by an honest node can decrease when a new checkpoint appears at a conflicting chain, thus, requiring a careful analysis to bound how many honest blocks are lost due to such instances.
However, it is not immediately obvious if the analysis of \cite{sankagiri_clc} that relies on \cite{backbone} carries over to the case of PoS protocols, where an adversary can generate multiple blocks when it is elected.
Hence, our goal is to show that the PoS protocols such as \cite{sleepy,david2018ouroboros,badertscher2018ouroboros} recover their safety and liveness after $O(\max(\GST, \GAT))$ time under $\AdvEnvPG$.
In this endeavor, we enhance the proof technique of \cite{sankagiri_clc} by using the concept of strong pivots from \cite{sleepy}.

In the proof below, we follow the same notation as \cite{sleepy}.
Each node is elected as the leader of a time slot with probability $p$.
Total number of slots $T_{\mathrm{max}}$ is fixed and is a polynomial in the security parameter $\sigma$.
There are $n$ nodes in total, among which $f$ nodes are controlled by the adversary.
We denote the checkpoint-respecting longest chain held by an honest node $i$ at slot $t$ by $\chlc{i}{t}$.

Define $\beta=pf$ as an upper bound on the expected number of adversary nodes elected leader in a single slot.
Similarly, define $\alpha$ as a lower bound on the expected number of awake honest nodes elected leader in a single slot. 
After $\GAT$, every honest node wakes up and $\alpha = p(n-f) > \beta$ as $f < n/2$.
For the convergence opportunities, we adopt the definition given in \cite[Section 5.2]{sleepy} and denote the number of convergence opportunities within a time interval $[t_1,t_2]$ by $\mathrm{C}([t_1,t_2])$.

We say that a block $b$ is checkpointed at slot $t$ if $t$ is the first slot an honest node sees $b$ as checkpointed.
Note that after $\GST$, if $b$ is checkpointed at slot $t$, then every honest node sees $b$ as checkpointed by slot $t+\Delta$.
Let $\max(\GST,\GAT)+\Trecent < t^*_1 \leq t^*_2 \leq ... \leq $ be the slots at which new blocks $b^*_1$, $b^*_2$, ... are checkpointed after $\max(\GST,\GAT)$.
To show that checkpointing a new block does not forfeit too many convergence opportunities, we follow the approach of \cite{sankagiri_clc} and divide time into two sets of intervals.
Let $I_l := [t^*_l + \Delta, t^*_{l+1} - \Trecent - \Delta]$ and define $I := \cup_{l \geq 0} I_l$ as the union of the \emph{inter-checkpoint intervals} $I_l$.
(Recall the definition of $\Trecent$ from Lemma~\ref{thm:checkpoint-recency}.) 
Using the definition of $I$, we can now proceed to prove the chain growth, quality and the common prefix properties.

\subsection{Chain Growth}
\label{sec:appendix-chain-growth}

\begin{definition}[{\cite[Section 3.2.1]{sleepy}}]
Predicate $\mathrm{growth}(\tau,k)$ is satisfied if and only if for every slot $t \leq T_{\mathrm{max}} - \tau$, $\min_{i,j}(|\chlc{j}{t+\tau}|-|\chlc{i}{t}|) \geq k$.
\end{definition}

We use the same results given in \cite{sankagiri_clc} to lower bound the chain growth in terms of convergence opportunities that lie within the inter-checkpoint intervals.

\begin{lemma}[{\cite[Lemma 5]{sankagiri_clc}}]
\label{thm:chain-growth-catch}
Let $s,t$ be two slots such that $s \in I$ and and $t \geq s + \Delta$. 
Let $\chlc{i}{s}$ be a chain held by some honest node $i$ at slot $s$. 
Then all honest nodes will hold a chain of length at least $|\chlc{i}{s}|$ in slot $t$.
\end{lemma}

Proof is (almost) the same as \cite[Lemma 5]{sankagiri_clc} and uses Lemma~\ref{thm:checkpoint-recency}.
\begin{proof}
Let $i$ and $j$ (potentially $i=j$) be honest nodes awake at slots $s$ and $t \geq s+\Delta$ respectively.
Since $s \in I$, there exists an $l$ such that $s \in I_{l}$. 
Let $b$ denote the last checkpoint block within $\chlc{i}{s}$. 
Since $s \geq \max(\GST,\GAT)$, all honest nodes accept $b$ as a checkpoint by slot $t$.
Next, consider the following two cases: (i) $b$ is the last checkpoint block in $j$'s view by slot $t$.
Then, $|\chlc{j}{t}| \geq |\chlc{i}{s}|$, as $\chlc{i}{s}$ contains all of the checkpoints observed by $j$ by slot $t$.
(ii) $j$ observes at least one new block become checkpointed by slot $t$.
In this case, let $b'$ denote the first block that becomes checkpointed in $j$'s view after $b$ within slot $t^*_{l'}$. $l'>l$.
(In this case, $t \geq t^*_{l'}$ by definition.)
Via Lemma~\ref{thm:checkpoint-recency} (recency property), $b'$ must be on the checkpoint-respecting LC $\chlc{k}{t'}$ held by an honest node $k$ at some slot $t' \in [t^*_{l'} - \Trecent, t^*_{l'}]$, during which $b'$ was not checkpointed yet in the view of any honest node.
Thus, via case (i), $|\chlc{k}{t'}| \geq |\chlc{i}{s}|$ since $t' \geq t^*_{l'} - \Trecent \geq s+\Delta$ for $l'>l$.
Note that the length of the checkpoint-respecting LC held by any honest node at the time it observes $b'$ become checkpointed must be at least $|\chlc{k}{t'}| \geq |\chlc{i}{s}|$.
Hence, by induction, we can state that all honest nodes that observe at least one new checkpoint (after $b$) by slot $t$ hold chains of length at least $|\chlc{i}{s}|$ at slot $t$, implying that $|\chlc{j}{t}| \geq |\chlc{i}{s}|$.
\end{proof}

A useful corollary of Lemma~\ref{thm:chain-growth-catch} is given below:
\begin{corollary}
\label{thm:unique}
All honest blocks produced at convergence opportunities within $I$ have distinct heights.
\end{corollary}
\begin{proof}
Suppose an honest block $b$ is produced at height $\ell$ at a convergence opportunity $t$ within $I$.
Then, the honest producer of $b$ holds a chain of length $\ell$ at slot $t \in S$.
Via Lemma~\ref{thm:chain-growth-catch}, all honest nodes will hold a chain of length at least $\ell$ at all slots $\geq t+\Delta$.
Since the next convergence opportunity after $t$ happens at some slot $\geq t+\Delta$, the next honest block will be at a height larger than $\ell$.
\end{proof}

\begin{lemma}[{\cite[Lemma 6]{sankagiri_clc}}]
\label{thm:chain-growth-lemma}
Consider a slot $s \in I$ and an honest node $i$ awake at $s$ such that $|\chlc{i}{s}|=\ell$.
(Alternatively, consider a slot $t$ and an honest node $i$ awake at $t$ such that $\chlc{i}{t}$ contains an honest block at height $\ell$ produced in some slot $s<t$.)
Then, for any slot $t \geq s+2\Delta$ and honest node $j$ awake at slot $t$, $|\chlc{j}{t}| \geq \ell + \mathrm{C}(I \cap [s+\Delta,t-\Delta])$.
\end{lemma}

Proof is similar to \cite[Lemma 6]{sankagiri_clc} and uses Lemma~\ref{thm:chain-growth-catch}.
A direct consequence of Lemma~\ref{thm:chain-growth-lemma} is that $\mathrm{growth}(\tau,k)$ is satisfied if for any interval $[t_1,t_2]$ of length $t_2-t_1 = \tau \geq 2\Delta$, $\mathrm{C}(I \cap [t_1+\Delta,t_2-\Delta]) \geq k$.

\subsection{Chain Quality}

\begin{definition}[{\cite[Section 3.2.2]{sleepy}}]
Predicate $\mathrm{quality}(\mu,k)$ is satisfied if and only if for every slot $t$ and every honest node $i$ awake at $t$, among any consecutive sequence of $k$ blocks within $\chlc{i}{t}$, the fraction of blocks produced by honest nodes is at least $\mu$.
\end{definition}

\begin{lemma}
\label{thm:chain-quality}
If $\mathrm{quality}(\mu,1/\mu)=0$, then there exist slots $s,t$ such that $\mathrm{A}([s,t]) \geq 1/\mu$ and $\mathrm{A}([s,t]) \geq \mathrm{C}(I \cap [s+\Delta,t-\Delta])-1$.
\end{lemma}

\begin{proof}
If $\mathrm{quality}(\mu,1/\mu)=0$, then there exists a slot $t'$ and an honest node $i$ awake at $t'$ such that $\chlc{i}{t'}$ contains a consecutive sequence of $1/\mu$ blocks $b_1,...,b_{1/\mu}$ produced by the adversary.
Let $b^*_s$ denote the last honest block before $b_1$ within $\chlc{i}{t'}$ and let $\ell_s$ and $s$ respectively denote its height and the slot it was produced in.
Similarly, let $b^*_t$ denote the first honest block after $b_{1/\mu}$ within $\chlc{i}{t'}$ and let $\ell_t$ and $t$ respectively denote its height and the slot it was produced in.
(Note that the genesis block can be taken as an honest block.)
If there is no honest block following $b_{1/\mu}$ within $\chlc{i}{t'}$, let $b^*_t = \chlc{i}{t'}[-1]$, $\ell_t = |\chlc{i}{t'}|$ and $t=t'$.
In either case, there exists an honest node which holds a chain of length $\ell_t-1$ at slot $t$ and this chain contains an honest block at height $\ell_s$ produced in slot $s$.
Thus, via Lemma~\ref{thm:chain-growth-lemma}, we know that $\ell_t \geq \ell_s + \mathrm{C}(I \cap [s+\Delta,t-\Delta])$.
Note that every block within $\chlc{i}{t'}$ with height in $(\ell_s,\ell_t)$ was produced by the adversary within the interval $[s,t]$, and lie on the same chain.
Hence, $\mathrm{A}([s,t]) \geq \ell_t-\ell_s-1$,
where $\ell_t-\ell_s-1 \geq \mathrm{C}(I \cap [s+\Delta,t-\Delta])-1$.
Moreover, the blocks $b_1,...,b_{1/\mu}$ were produced within the interval $[s,t]$ and lie on the same chain, implying that $\mathrm{A}([s,t]) \geq 1/\mu$.
Hence, we can conclude that if $\mathrm{quality}(\mu,1/\mu)=0$, $\mathrm{A}([s,t]) \geq 1/\mu$ and $\mathrm{A}([s,t]) \geq \mathrm{C}(I \cap [s+\Delta,t-\Delta])-1$.
\end{proof}

\subsection{Common Prefix}
\label{sec:appendix-common-prefix}

\begin{definition}[{\cite[Section 3.2.3]{sleepy}}]
Predicate $\mathrm{prefix}(\tau)$ is satisfied if and only if for all slots $s \leq t$ and honest nodes $i,j$ such that $i$ and $j$ are awake at slots $s$ and $t$ respectively, prefix of $\chlc{t}{j}$ consisting of blocks produced at slots $\leq t-\tau$ is a prefix of $\chlc{s}{i}$.
\end{definition}

To show the common prefix property in the context of checkpointed PoS protocols, we extend the definition of strong pivots in \cite[Section 5.6.1]{sleepy} as shown below:

\begin{definition}
A slot $t$ is said to be a \emph{checkpoint-strong pivot}, if for any $t_0 \leq t \leq t_1$, it holds that either $\mathrm{A}([t_0,t_1]) < \mathrm{C}(I \cap [t_0+\Delta,t_1-\Delta])$ or $\mathrm{A}([t_0,t_1])=0$.
\end{definition}

Observe that when we count the number of convergence opportunities for a checkpoint-strong pivot, we only take those that lie within the inter-checkpoint intervals. 
Intuitively, this is because the convergence opportunities that arrive during checkpoint intervals do not offer any guarantee of growth for the chains held by honest nodes.
Conversely, as Corollary~\ref{thm:unique} states, all of the honest blocks that arrive at convergence opportunities within $I$ have a unique height.
Hence, by counting only the convergence opportunities in $I$, we can inherit all of the qualitative results presented in \cite{sleepy} about the prefix property.
In this context, following proposition and lemma extend \cite[Fact 4, Lemma 5]{sleepy}:

\begin{proposition}[Unique Honest Blocks at Convergence Opportunities in $I$]
\label{thm:common-prefix-prop}
Let $i$ and $j$ be two honest nodes awake at slots $r_1$ and $r_2\geq r_1$ respectively. 
If $\chlc{i}{r_1}[\ell]$ and $\chlc{j}{r_2}[\ell]$ are both honest blocks and there exists a convergence opportunity $t^*$, $t^* \in I$, such that an honest block $b^*$ was produced at height $\ell$, then, $\chlc{i}{r_1}[\ell]=\chlc{j}{r_2}[\ell]=b^*$.
\end{proposition}

\begin{proof}
For the sake of contradiction, suppose $\chlc{i}{r_1}[\ell]$ and $\chlc{j}{r_2}[\ell]$ are both honest blocks at least one of which is different from $b^*$.
Let $k$ denote the honest producer of $b^*$ such that $\chlc{k}{t^*}[\ell]=b^*$.
Without loss of generality, suppose $\chlc{i}{r_1}[\ell]=b \neq b^*$, and let $m$ and $t$ denote the honest block producer and the production slot of $b$.
As $b \neq b^*$, either $t < t^*-\Delta$ or $t > t^*+\Delta$.
Now, if $t<t^*-\Delta$, either at least one honest node holds a checkpoint-respecting LC of length $\ell$ at time $t^*-\Delta$, or $b$ conflicts with one of the blocks checkpointed before slot $t^*-\Delta$.
In the first case, $|\chlc{k}{t^{*-}}| \geq \ell$, which implies $b^*$ could not have been produced at height $\ell$, leading to a contradiction.
In the latter case, no checkpoint-respecting LC of an honest node will contain $b$ after slot $t^*$, which is a contradiction as $b \in \chlc{i}{r_1}$.
Conversely, if $t>t^*+\Delta$, via Lemma~\ref{thm:chain-growth-catch}, $|\chlc{m}{t^-}| \geq |\chlc{k}{t^*}| \geq \ell$, which implies $b$ could not have been produced at height $\ell$, leading to a contradiction.
Thus $b=b^*$ and $\chlc{i}{r_1}[\ell]=\chlc{j}{r_2}[\ell]=b^*$.
\end{proof}

\begin{lemma}
\label{thm:checkpoint-strong-pivot}
Let $i$ and $j$ be two honest nodes awake at slots $r_1$ and $r_2\geq r_1$ respectively. 
Let $t$ be a checkpoint-strong pivot such that there is a convergence opportunity in $[t,r_1] \cap I$. 
Then, the last common block within $\chlc{i}{r_1}$ and $\chlc{j}{r_2}$ should have been produced in a slot $\geq t$.
\end{lemma}

\begin{proof}
Since $t$ is a checkpoint-strong pivot, there exist convergence opportunities $t'' \leq t \leq t'$ such that $t'',t' \in I$ and no adversary block is produced in the interval $[t'',t']$.
(Note that if $t$ is a convergence opportunity in $I$, $t''=t'=t$.)
In any case, $t'$ is also a checkpoint-strong pivot.
By the assumption that there is a convergence opportunity in $[t,r_1] \cap I$, we know that $t' \leq r_1$.

As $t' \in I$ is a convergence opportunity, there exists an honest block $b^*$ produced at some height $\ell^*$ at slot $t'$.
Next, we claim that since $t'$ is a checkpoint-strong pivot, there cannot be an adversarial block within $\chlc{i}{r_1}$ and $\chlc{j}{r_2}$ at height $\ell^*$.
For the sake of contradiction, suppose there exists an adversary block $b$ at height $\ell^*$ in a chain $\chlc{m}{s}$ held by some honest node $m$ at some slot $s \geq t'+\Delta$.
Let $b_1$, produced at slot $t_1$ and height $\ell_1<\ell^*$, denote the last honest block in the prefix of $b$ within $\chlc{m}{s}$.
If there is no such block, we take the genesis block as $b_1$.
Similarly, let $b_2$, produced at slot $t_2$ and height $\ell_2>\ell^*$, denote the first honest block in the suffix of $b$ within $\chlc{m}{s}$.
If there is no such block, we take the last block on $\chlc{m}{s}$ as $b_2$.
If $b_2$ is an honest block, by definition, all of the blocks on $\chlc{m}{s}$ at heights $(\ell_1,\ell_2)$ are adversarial (there is at least one such block $b$) and have been mined at distinct times within the interval $(t_1,t_2)$.
On the other hand, if $b_2$ is adversarial, all of the blocks on $\chlc{m}{s}$ at heights $(\ell_1,\ell_2]$ are adversarial and have been mined at distinct times within the interval $(t_1,t_2]$.
In this context, we first consider the case that $b_2$ is honest.

Observe that as $t'$ is a convergence opportunity, either $t' \leq t_1-\Delta$, $t' \in [t_1+\Delta,t_2-\Delta]$ or $t' \geq t_2+\Delta$.
First, if $t' \leq t_1-\Delta$, via Lemma~\ref{thm:chain-growth-catch}, all honest nodes will hold a chain of length at least $\ell^*$ by time $t_1$, implying that $b_1$ could not have been mined at height $\ell_1<\ell^*$, which is a contradiction.
Second, if $t' \geq t_2+\Delta$, then the honest producer of $b^*$ has already seen $b_2$; yet decided to produce a block at height $\ell^*<\ell_2$.
However, this is only possible if $b_2$ conflicts with a checkpoint observed by the producer of $b^*$ by time $t'$.
However, this checkpoint would also be seen by node $m$ by time $t'+\Delta \leq s$, implying that $b_2$ cannot be on $\chlc{m}{s}$ at time $s$, which again is a contradiction.
Consequently, the only possible scenario for $t'$ is $t' \in [t_1+\Delta,t_2-\Delta]$. 

Finally, suppose $b_2$ was produced by some honest node $m'$.
As $\chlc{m'}{t_2}$ contains an honest block, $b_1$, at height $\ell_1$ produced at slot $t_1<t_2-\Delta$, via Lemma~\ref{thm:chain-growth-lemma}, $\ell_2=|\chlc{m'}{t_2}|>\ell_1+\mathrm{C}(I \cap [t_1+\Delta,t_2-\Delta])$.
Moreover, all of the blocks between $b_1$ and $b_2$ (perhaps including $b_2$) are adversarial and has been mined in the interval $(t_1,t_2]$ at distinct times, implying that $A([t_1,t_2]) \geq \mathrm{C}(I \cap [t_1+\Delta,t_2-\Delta])$.
However, this is a contradiction with the fact that $t' \in [t_1+\Delta,t_2-\Delta]$ is a checkpoint-strong pivot.

Next, we consider the case that $b_2$ is an adversarial.
From the analysis above, we know that $t'>t_1+\Delta$.
Moreover, every block on $\chlc{m}{s}$ following $b_1$ are adversarial blocks.
As the honestly-held chain $\chlc{m}{s}$, $s > t'+\Delta$, contains an honest block, $b_1$, at height $\ell_1$ produced at slot $t_1<s-\Delta$, via Lemma~\ref{thm:chain-growth-lemma}, $\ell_2=|\chlc{m}{s}|>\ell_1+\mathrm{C}(I \cap [t_1+\Delta,s-\Delta])$.
Moreover, all of the blocks between $b_1$ and $b_2$ (including $b_2$) are adversarial and have been mined in the interval $(t_1,s]$ at distinct times, implying that $A([t_1,s]) \geq \mathrm{C}(I \cap [t_1+\Delta,s-\Delta])$.
However, this is a contradiction with the fact that $t' \in [t_1+\Delta,s-\Delta]$ is a checkpoint-strong pivot.

Consequently, there cannot exist an adversary block $b$ at height $\ell$ in any chain held by honest nodes after slot $t'+\Delta$, an assertion that includes the chains $\chlc{i}{r_1}$ and $\chlc{j}{r_2}$.
Then, as $\chlc{i}{r_1}[\ell^*]$ and $\chlc{j}{r_2}[\ell^*]$ are both honest blocks and $t'$ is a convergence opportunity, via Proposition~\ref{thm:common-prefix-prop}, \begin{IEEEeqnarray}{C}\chlc{i}{r_1}[\ell^*]=\chlc{j}{r_2}[\ell^*]=b^*.\end{IEEEeqnarray}
Hence, the last common block within  $\chlc{i}{r_1}$ and $\chlc{j}{r_2}$ must have been produced in a slot $\geq t' \geq t$.
\end{proof}

\subsection{Probabilistic Analysis}
\label{sec:appendix-probabilistic-analysis}

To lower bound the number of convergence opportunities within $I$, we can use the following observation from \cite[Proposition 4]{sankagiri_clc} which relies on Proposition~\ref{thm:checkpoint-gap} (gap property): If $t \geq s \geq \max(\GST,\GAT)$, \begin{IEEEeqnarray}{C}
\mathrm{C}([t_1+\Delta,t_2-\Delta] \cap I) \geq \mathrm{C}([t_1+\Delta,t_2-\Delta]) - (1+(t_2-t_1))(\Trecent+2\Delta+1)/\Tcheckpoint
\IEEEeqnarraynumspace
\end{IEEEeqnarray}
Combining this expression with \cite[Lemma 2, Corollary 2, Fact 2]{sleepy} yields the following proposition:

\begin{proposition}
\label{thm:number-of-c}
For any $\epsilon>0$, there exists an $\epsilon'$ such that given $t_2 \geq t_1 \geq \max(\GST,\GAT)$, $t\triangleq t_2-t_1 \geq 2\Delta$, 
\begin{IEEEeqnarray*}{rCl}
&& \mathbb{P}\left[\mathrm{C}([t_1+\Delta,t_2-\Delta] \cap I) \leq \left((1-\epsilon)(1-2pn\Delta)\alpha-\frac{\Trecent+2\Delta+1}{\Tcheckpoint}\right)t\right] \\
&<& \exp{(-\epsilon'\alpha t)} \\
&& \mathbb{P}[\mathrm{A}([t_1,t_2]) > (1+\epsilon)\beta t] < \exp{(-\epsilon^2\beta t/3)}.
\end{IEEEeqnarray*}
\end{proposition}

We also note that for any given $p<(n-2f)/(2\Delta n(n-f))$ and sufficiently large $\Tcheckpoint$, there exists a constant $\epsilon>0$ such that
\begin{equation}
\label{eq:ineq}
    (1+\epsilon)\beta<(1-\epsilon)(1-2pn\Delta)\alpha-\frac{\Trecent+2\Delta+1}{\Tcheckpoint}
\end{equation}

Next, we define $T$ as the minimum slot $t \geq \max(\GST,\GAT)$ such that $\mathrm{C}([0, t-\Delta] \cap I)=\mathrm{A}([0,t])$.
In other words, $T$ is an upper bound on the slot by which checkpoint-respecting LCs held by honest nodes would have caught up with the checkpoint-respecting LCs held by the adversary nodes.
Thus, we can view $T$ as the time $\PIlc$ resets itself such that after $T$, it behaves like a checkpoint-respecting LC protocol that has just started running in a synchronous network.
As long as $\Tcheckpoint$ is selected sufficiently large for equation~\ref{eq:ineq} to hold, combining \cite[Propositions 2,3,4]{ebbandflow} with Proposition~\ref{thm:number-of-c}, we can assert the following proposition bounding $T$: 
\begin{proposition}
\label{thm:catch-up}
There exists a constant $\mathbf{C}$ such that for any given security parameter $\sigma$, and $\GST,\GAT$ specified by $\AdvEnvPG$, $T \leq \mathbf{C}(\max(\GST,\GAT) + \sigma)$ except with probability $e^{-\Omega(\sigma)}$.
\end{proposition}
Using Proposition~\ref{thm:catch-up}, we can complete the proof of Theorem~\ref{thm:security-of-PIlc}:
\begin{proof}
Using Proposition~\ref{thm:number-of-c} and Lemma~\ref{thm:chain-growth-lemma}, we can assert that for any given $\epsilon>0$, $\mathrm{growth}(\sigma, k)$ is satisfied after $\max(\GST,\GAT)$ except with probability $e^{-\Omega(\sigma)}$, where $k=g_0 \sigma$ for \begin{IEEEeqnarray}{C}g_0=(1-\epsilon)(1-2pn\Delta)\alpha-\frac{\Trecent+2\Delta+1}{\Tcheckpoint}.\end{IEEEeqnarray}
Similarly, using Lemma~\ref{thm:chain-quality}, Proposition~\ref{thm:number-of-c} and Proposition~\ref{thm:catch-up}, we can assert that for any given $\epsilon>0$, $\mathrm{quality}(\mu,1/\mu)$ is satisfied after slot $\mathbf{C}(\max(\GST,\GAT)+\sigma)$, except with probability $e^{-\Omega(\sigma)}$, where 
\begin{IEEEeqnarray}{C}\mu=\frac{(1-\epsilon)(1-2pn\Delta)\alpha-(1+\epsilon)\beta-(\Trecent+2\Delta+1)/\Tcheckpoint}{(1-2pn\Delta)\alpha-(\Trecent+2\Delta+1)/\Tcheckpoint}.\end{IEEEeqnarray}
Finally, we know via Lemma~\ref{thm:checkpoint-strong-pivot} that checkpoint-strong pivots force convergence of the checkpoint-respecting LCs seen by all honest nodes.
Hence, we can use \cite[Theorem 7]{sleepy} to show that $\mathrm{prefix}(\sigma)$ is satisfied after slot $\mathbf{C}(\max(\GST,\GAT)+\sigma)$ except with probability $e^{-\Omega(\sqrt{\sigma})}$.
Then, using \cite[Lemma 1]{sleepy}, we conclude that $\PIlc$ is secure with confirmation time $O(\sigma/g_0)$ 
after slot $\mathbf{C}(\max(\GST,\GAT)+\sigma)$ except with probability $e^{-\Omega(\sqrt{\sigma})}$.
\end{proof}

Via Lemma~\ref{thm:checkpoint-recency}, $\Trecent = \Delta + \Ttimeout + \Tconfirm$.
On the other hand, $\Tcheckpoint$ should be chosen large enough for the inequality~\eqref{eq:ineq} to be satisfied.
Thus, placing the definitions of $\alpha$ and $\beta$, setting $f = n/4$,  $\Ttimeout = 60$ seconds as in Section~\ref{sec:experiments}, and choosing $\epsilon = 0.1$ and $p=0.8(n-2f)/(2\Delta n(n-f))=0.8/(3n\Delta)$, we can obtain inequality~\eqref{eq:ineq} from the following expression:
\begin{equation}
    \frac{2\Ttimeout+3\Delta}{(1+\epsilon)pf - (1-\epsilon)(1-2pn\Delta)p(n-f)2\Ttimeout} = 100(120+3\Delta)\Delta < \Tcheckpoint
\end{equation}
Thus, for any given value of $\Delta$, there exists a $\Tcheckpoint$ that satisfies inequality~\eqref{eq:ineq} for these set of parameters.
Experimentation in Section~\ref{sec:experiments} shows that for $\Delta$ approximately a few seconds, $\Tcheckpoint$ of $300$ seconds is sufficiently large to ensure security under real network conditions.

\subsection{Security Argument for Chia}
\label{sec:appendix-pospace}

While the sections above prove Theorem~\ref{thm:main-security} for PoS, and by implication PoW protocols, security of Chia \cite{cohen2019chia} does not immediately follow from the analysis of checkpoint-strong-pivots due to nothing-at-stake attacks \cite{dem20}, which enable the adversary to mine blocks on top of each existing block via independent Poisson processes.
The first paper to show security for such protocols given the possibility of nothing-at-stake attacks is \cite{dem20} which introduced a novel method called blocktree-partitioning.
This method splits the overall blocktree into adversarial trees that build on a \emph{fictitious honest tree} with a chain growth property analogous to the one in Appendix~\ref{sec:appendix-chain-growth}.
Thus, as in the case of convergence opportunities, we can once again count only the honest blocks that arrive within inter-checkpoint intervals $I_l$ to provide a non-trivial lower bound on the growth of the fictitious honest tree in the context of checkpoint-respecting LCs.
This lower bound follows from the fact that each honest block produced during such an interval $I_l$ has the potential to contribute to the growth of honest chains just like in the case of original LC protocols.
Using the modified definition for the fictitious honest tree, we can then prove \cite[Theorem 3.2]{dem20} that ties protocol security to the evolution of the fictitious honest tree for checkpoint-respecting LC protocols.
Finally, the probabilistic analysis of \cite[Section 4.2]{dem20} goes through provided that $\beta<1/e$ and the parameters $p$ and $1/\Tcheckpoint$ are sufficiently small.
Details of this analysis is left as future work. 

\section{Experimental Evaluation Details}
\label{sec:appendix-experiments}

\paragraph{Implementation:}
Our prototype is implemented in the programming language Rust.
A diagram of the different components and their interactions is provided in Figure~\ref{fig:protocol-systems-diagram}.
We use a longest chain protocol modified to respect latest checkpoints as $\PIlc$, with a permissioned block production lottery with winning probability $p$ per node and per time slot of duration $\Tslot$; and
HotStuff\footnote{We used this Rust implementation: \url{https://github.com/asonnino/hotstuff}} as $\PIbft$. Honest nodes pause HotStuff (including its timeouts) while waiting for the next checkpoint proposal.
All communication (including HotStuff's) takes place in a broadcast fashion via \texttt{libp2p}'s \texttt{Gossipsub} protocol\footnote{We used this Rust implementation: \url{https://github.com/libp2p/rust-libp2p}}, mimicking Ethereum 2's network layer \cite{eth2-spec-p2p}, to be able to scale to thousands of nodes.
Thus, we assume that under normal conditions every message
received by one honest node
will be received by all honest nodes within some bounded delay.
Since responsiveness is not so important for our checkpointing application
and to avoid broadcasting quorum certificates,
we use a variant of HotStuff where
to ensure liveness
the leader waits
for the network delay bound before proposing a block.
Our prototype does not implement the application logic of the beacon chain (such as validators joining and leaving, integration with shard chains and Ethereum 1, etc.) which can be realized on top of consensus in the same way as currently done in Ethereum 2,
and our prototype does not use any orthogonal techniques to reduce bandwidth by constant factors (such as signature aggregation, short signature schemes, compression of network communication, etc.) which are not fundamental to the consensus problem.

\paragraph{Choice of parameters:}
We chose the parameters of \ourprotocol in the experiments to match the parameters of Ethereum 2's beacon chain.
The beacon chain has $C=32$ slots per epoch and $m=128$ validators per slot, for a total of $n=4096$ validators (per epoch), which is the approximate number of nodes that we run our experiments with.
To match the block inter-arrival time (\ie, the duration of one slot) of $12\,\mathrm{s}$ in the beacon chain, we set $p = 1/n$ and account for the probability of no node winning the block production lottery and choose $\Tslot = 7.5\,\mathrm{s}$.
We also match the block payload size of $22\,\mathrm{KBytes}$.
In terms of $\PIlc$, we chose $\kcp=6$ so that a checkpoint proposal by an honest leader is reasonably likely (although not `guaranteed') to be accepted by other honest nodes, and $k=6$ for a swift $72\,\mathrm{s}$ average confirmation delay of $\LOGda{}{}$. Note that $\kcp$ should be the same for all nodes, while each client can choose an individual $k$ to trade off latency and confirmation error probability of $\LOGda{}{}$.
To leave enough time for message propagation,
we set the HotStuff timeout $\Thotstuff = 20\,\mathrm{s}$.
To avoid HotStuff timeouts escalating into checkpoint timeouts
for honest leaders, we set $\Ttimeout = 1\,\mathrm{min}$.
Finally, to target $5\times$ improvement in average
$\LOGacc{}{}$ latency over Gasper (\cf Figure~\ref{fig:latency-bandwidth}), we set $\Tcheckpoint = 5\,\mathrm{min}$.

\paragraph{Experiment setup:}
Adversarial nodes in the experiment
aim to stall consensus as much as possible.
Thus, they do not
share a proposal when elected leader in $\PIbft$ or $\PIacc$, so that honest nodes
have to wait for a timeout before they can move on,
and they mine selfishly \cite{selfishmining} in $\PIlc$
to reduce honest chain growth.
We ran our prototype
(a) with no adversary (Figure~\ref{fig:prototype-results-ledgers-nofaults}(l)),
and (b) with $\beta=25\%$ adversary (Figure~\ref{fig:prototype-results-ledgers-faults}(r)),
each 
for $2500\,\mathrm{s}$ on five AWS EC2 \texttt{c5a.8xlarge} instances in each of ten AWS regions\footnote{\texttt{eu-north-1}, \texttt{eu-west-3}, \texttt{ap-south-1}, \texttt{ap-northeast-1}, \texttt{ap-southeast-2}, \texttt{sa-east-1}, \texttt{ca-central-1}, \texttt{us-west-1}, \texttt{us-east-2}, \texttt{us-east-1}}, with $82$ nodes per machine, for a total of $4100$ nodes.
Each honest (adversarial) node connected to $15$ ($15$ honest and $15$ adversarial) randomly selected peers for the
peer-to-peer
gossip network.

\paragraph{Observations:}
We observe that both without faults
(Figure~\ref{fig:prototype-results-ledgers-nofaults}(l))
as well as under the $25\%$-attack
(Figure~\ref{fig:prototype-results-ledgers-faults}(r))
the available full ledger (\ref{leg:prototype-results-ledgers-nofaults-LOGava})
shows steady growth, albeit under attack at a reduced rate
due to selfish mining.
In both cases,
the accountable prefix ledger (\ref{leg:prototype-results-ledgers-nofaults-LOGacc})
periodically catches up with the available ledger.
Timeouts cause occasional but overall minor
delays of the catch-up.

In terms of bandwidth
(reported in Figures~\ref{fig:prototype-results-traffic-faults} and \ref{fig:prototype-results-traffic-nofaults} \emph{for an exemplary AWS instance},
\ie, for $82$ nodes),
we observe a distinct spiky pattern
with frequent small spikes corresponding to the propagation
of $\PIlc$ blocks and infrequent wide spikes corresponding
to the propagation of checkpoint votes and $\PIbft$ blocks and votes as part of checkpointing.
There is more traffic under attack
than without attack
(per node:
avg. $78\,\mathrm{KB/s}$ peak $1.5\,\mathrm{MB/s}$
vs avg. $56\,\mathrm{KB/s}$ peak $1.34\,\mathrm{MB/s}$),
since timeouts due to adversarial non-action
lead to additional iterations of checkpointing and HotStuff.
In either case,
the bandwidth requirement does not pose
a severe limitation
to participation
even using
consumer-grade Internet connectivity.

\paragraph{Bandwidth requirement and accountable ledger latency:}
We examine the tradeoff between the average number of votes
communicated per time
(as a surrogate for average required bandwidth,
to avoid confounding factors such as compression
or signature aggregation)
and the average latency of the accountable ledger (see Figure~\ref{fig:latency-bandwidth}), for varying number $n$ of nodes
and varying parameters $C$ and $\Tcheckpoint$ for Gasper and \ourprotocol,
respectively, under ideal operation,
\ie, $\beta = 0, \Delta = 0$.
In this case, Gasper transmits
$2 \cdot \frac{n}{C}$ votes per $12\,\mathrm{s}$ (per slot, each
committee member issues an LMD GHOST and a Casper FFG vote),
while \ourprotocol transmits $5 \cdot n$ votes per $\Tcheckpoint$ time
(broadcast checkpoint votes, checkpoint votes
in HotStuff proposal, three rounds of HotStuff voting for confirmation).
A transaction takes on average $\frac{1}{2} + 2$ epochs to enter
into the accountable ledger for Gasper (wait until end of ongoing epoch, then two epochs to reach finality),
and $\kcp \cdot 12\,\mathrm{s} + \frac{1}{2}\cdot\Tcheckpoint$
time to enter into the accountable ledger for \ourprotocol
($\kcp$-deep to enter checkpoint proposal, then wait until next checkpoint iteration).
As evident from Figure~\ref{fig:latency-bandwidth},
\ourprotocol offers slightly improved latency at comparable
bandwidth, or comparable bandwidth and latency but for
a larger number of nodes.
Let us point out that, as currently implemented,
nodes in \ourprotocol broadcast votes at the
highest throughput feasible once
$\Tcheckpoint$ has expired, so that the resulting
traffic pattern is more bursty
than that produced by Gasper,
where voting is taking place throughout each epoch.
However, Figure~\ref{fig:latency-bandwidth} also corroborates that
even if voting after $\Tcheckpoint$
was artificially rate-limited, bandwidth and latency comparable
to Gasper can be achieved.

Note that the gross bandwidth
measured in Figures~\ref{fig:prototype-results-traffic-nofaults}
and \ref{fig:prototype-results-traffic-faults} is roughly $6\times$ the bandwidth
estimate based on the number of votes per time.
This is largely due to two factors:
\emph{1)} We have not optimized HotStuff in our prototype
to remove quorum certificates from blocks
(although we may do so due to the broadcast nature of the
gossip network, as discussed in the earlier `Implementation' paragraph),
\emph{2)} the amplification factor that comes with nodes flooding every new message to all their peers in the gossip network.

\paragraph{Scaling the number of nodes for fixed accountable ledger latency:}
To scale the number of nodes for a fixed accountable ledger latency
(and hence fixed $C$ and $\Tcheckpoint$, respectively),
both Gasper and \ourprotocol need to increase their vote bandwidth
proportionally.
However, the attack described in
\cite{ethresearch-balancing-attack2}
suggests that even without adversarial but with merely
random network delay, Gasper is susceptible to
balancing by an adversary controlling
$O(\sqrt{C/n})$ of nodes.
Thus,
for fixed $C$, the relative adversarial resilience
decreases (to $0$) as $n$ increases (to $\infty$).

\paragraph{Improvements:}
Obvious improvements would be to customize the block proposal
generation of HotStuff to account for the semantics of votes and the state of the checkpointing protocol, \eg, do not propose votes from past checkpoint iterations, delay proposals when there are no pending votes, propose blocks with the minimum set of votes to reach a new checkpoint decision, etc.

\section{Helper Functions for Algorithms~\ref{algo:pseudocode-checkpoint-proposer} and~\ref{algo:pseudocode-checkpoint-extractor}}
\label{sec:appendix-pseudocode-helpers}

\begin{itemize}
    \item $\Call{PerformBookkeeping}{}$:
        Macro for bookkeeping in Alg.~\ref{algo:pseudocode-checkpoint-proposer}
        of checkpoint decisions
        and checkpoint proposals from checkpoint leaders in charge
    \item $\operatorname{CpLeaderOfIter}(c)$:
        Returns randomly selected publicly verifiable unique leader $\ld{c}$
        of checkpoint iteration $c$
    \item $\Call{Broadcast}{}(...)$:
        Broadcasts checkpoint proposal to other nodes
    \item $\Call{GetCurrProposalTip}{\null}$:
        Returns $\kcp$-deep block on node's checkpoint-respecting LC
        from $\PIlc$
    \item $\Call{IsValidProposal}{b}$:
        Checks whether block $b$ is consistent with
        $\kcp$-deep block on node's checkpoint-respecting LC
        as obtained from $\PIlc$
    \item $\Call{SubmitVote}{}(v)$:
        Inputs vote $v$ as payload to $\PIbft$ for ordering
    \item $\Call{GetNextVerifiedVoteFromBft}{\null}$:
        Retrieves next vote with valid signature
        from the output ledger $\LOGbft{}{}$
        as ordered by $\PIbft$
    \item $\Call{OutputCp}{...}$:
        Alg.~\ref{algo:pseudocode-checkpoint-proposer} outputs
        a new checkpoint decision
        (used as input in Alg.~\ref{algo:pseudocode-checkpoint-extractor}
        and to determine checkpoint-respecting LC in $\PIlc$)
\end{itemize}

\section{Proof-of-Work and Proof-of-Space}
\label{sec:appendix-extension-to-PoW}

We have so far focused on accountability gadgets built for permissioned LC protocols.
We conclude with an outlook on
accountability gadgets for permissionless LC protocols
such as those based on Proof-of-Work (PoW), \eg, Bitcoin \cite{nakamoto_paper}, or those based on Proof-of-Space (PoSpace), \eg, Chia \cite{cohen2019chia}.
In such constructions,
nodes fulfill two different roles:
\emph{miners} control a unit of 
a rate-limiting resource, \eg, compute power or storage space, and are responsible for extending the permissionless checkpoint-respecting LC; and \emph{validators} with unique cryptographic identities are responsible for providing accountability.
Security of $\PIlc$ is primarily maintained by miners, 
security of $\PIacc$ and the associated protocol $\PIbft$ are primarily maintained by validators.
While miners are free to participate dynamically, validators are expected to be always present to provide accountability.

In the following, let
$\beta$ be the fraction of online rate-limiting resource controlled by adversarial miners,
and $\beta^*$ be a quantity depending on the underlying $\PIlc$ protocol:
$\beta^* = 1/2$ for Bitcoin and $\beta^* = 1/e$ for Chia,
our two examples.
Then, the accountable ledger $\LOGacc{}{}$ and the available ledger $\LOGda{}{}$ satisfy:

Consider a network with
validators (with $n$ unique cryptographic identities)
and miners (at least one of which is honest and awake), with the properties discussed above.
The network may partition, and the miners may come online and go offline subject to the constraints below.
Then, for any $f \leq \lceil n/2 \rceil$:
\begin{enumerate}
    \item (\textbf{P1: Accountability}) The accountable ledger $\LOGacc{}{}$ provides \asr $(n-2f+2)$ at all times, and is live after network partition, \emph{if $\beta < \beta^*$ holds}, and if the number of honest validators is greater than $(n-f)$.
    \item (\textbf{P2: Dynamic Availability}) The available ledger $\LOGda{}{}$ is guaranteed to be safe after network partition and live at all times, if $\beta < \beta^*$ holds, and \emph{if the number of adversarial validators is at most $(n-f)$}.
\end{enumerate}
Again, $\LOGacc{}{}$ is always a prefix of $\LOGda{}{}$ by construction.

Observe that the requirements of having greater than $(n-f)$ honest nodes under \textbf{P1} and having $\beta < \beta^*$ under \textbf{P2} are analogous to the permissioned case formulated in
Theorem~\ref{thm:main-security}.
However, accountability gadgets for permissionless LC protocols have further requirements.
First, under \textbf{P1}, liveness of $\LOGacc{}{}$ requires $\beta < \beta^*$.
Otherwise, the 
adversarial miners might stall $\PIlc$ and hence $\LOGda{}{}$ might not become live after network partition, which was argued to be a necessary condition for liveness of $\LOGacc{}{}$ in Section~\ref{sec:proof-sketch}.
Second, under \textbf{P2}, security of $\LOGda{}{}$ requires the number of adversarial validators to be bounded by $(n-f)$.
Otherwise, if the number of adversarial validators was greater than $(n-f)$, it would be possible for a block that is not on a checkpoint-respecting LC held by any honest nodes to become checkpointed solely by adversarial votes (l.~\ref{line:threshold1} of Alg.~\ref{algo:pseudocode-checkpoint-extractor}), thus causing safety violations on $\LOGda{}{}$.

For Bitcoin,
the security proof of
Appendices~\ref{sec:appendix-security-proofs} and~\ref{sec:appendix-security-proof-for-the-longest-chain-protocol}
applies.
For Chia, security can be proven via an extension of 
\emph{blocktree partitioning} \cite{dem20} to 
checkpoint-respecting LCs.
Details in
Appendix~\ref{sec:appendix-pospace}.
Finally, note that the reduced threshold $\beta^* = 1/e$ for Chia as the permissionless LC protocol
is due to nothing-at-stake attacks \cite{dem20}, since randomness is updated per
block in Chia.
Other embodiments of PoSpace may provide different $\beta^*$.

\section{Proof of Non-Accountability of Checkpointed LC}

\label{sec:appendix-non-accountability-clc}

In this section, we show that the checkpointed longest chain protocol presented in \cite{sankagiri_clc} does not provide accountable safety.
The checkpointing protocol used by \cite{sankagiri_clc} is a slight modification of the Algorand BA protocol from \cite{algorand_agreement}.
Thus, the attack below on the accountability of the Algorand BA in \cite{sankagiri_clc} is very similar to the one described for Algorand BBA* \cite{algorand} in \cite[Appendix C.3]{forensics}.

\begin{theorem}
Algorand BA \cite{algorand_agreement} does not provide accountable safety.
\end{theorem}

\begin{proof}
For the sake of contradiction, suppose there exists an adjudication function $\adj$ that can identify at least one adversary node when there is a safety violation, and never identifies an honest node.
Let $n=3f+1$ denote the total number of nodes among which $f$ nodes are controlled by a Byzantine adversary. 
Note that Algorand BA requires each node $i$ to hold a \emph{starting value} $st_i^p$ (that is distinct from the input values) for each period $p$ of the protocol.
Starting values are set to $st_i^1 = \perp$ for period $p=1$.
Each period consists of $4$ steps and an optional $5$-th step.

Consider three disjoint sets of nodes $P$, $Q$ and $R$, where $|P|=f-1$ and $|Q|=|R|=f+1$.
We next construct two worlds each with a different set of Byzantine nodes. 

\textbf{World 1:} Nodes in $R$ are controlled by the adversary. 
Suppose that the nodes in $P$ and $Q$ start with the input bits $0$ and $1$ respectively, and a node from $P$ is elected leader at period $p=1$.
Then, the following steps are executed by Algorand BA.
\begin{itemize}
    \item Period $1$, Step $1$: Leader proposes its input, $0$.
    \item Period $1$, Step $2$: Every node soft-votes the value $0$ proposed by the leader. 
    Adversary nodes also soft-vote $0$ and share their votes with all honest nodes.
    \item Period $1$, Step $3$: Nodes in $P$ and $Q$ see more than $2f+1$ soft-votes for $0$, thus they cert-vote for $0$.
    Nodes in $R$ also cert-vote for $0$, but send their cert-votes only to the nodes in $P$.
    \item Period $1$, Step $4$: Nodes in $P$ receive $3f+1>2f+1$ cert-votes, thus terminate using the halting condition and output $0$.
    Nodes in $Q$ receive only $|P|+|Q|=2f$ cert-votes, thus are not able to certify any value.
    Nodes in $R$ pretend as if they do not receive any cert-votes from the nodes in $Q$, thus are not able to certify any value either.
    Hence, nodes in $Q$ and $R$ go to the period's first finishing step. 
    They all next-vote their starting values $st_i^1$, which is set to $\perp$ for period $p=1$.
    Note that the nodes in $R$ send their next-votes to the nodes in $Q$.
    \item Period $2$, Step $1$: Since the nodes in $Q$ and $R$ observe a total of $2f+1$ $\perp$-next-votes, they do not go to the second finishing step of period $1$ and instead jump to step $1$ of period $p=2$, with $st_i^2 = \bot$.
    Suppose a node from $Q$ is elected leader at period $p=2$.
    It proposes its input $1$.
    \item Period $2$, Step $2$: Nodes in $Q$ and $R$ soft-vote the value $1$ proposed by the leader.
    \item Period $2$, Step $3$: Nodes in $Q$ and $R$ see more than $2f+1$ soft-votes for $1$, thus they cert-vote for $1$.
    \item Period $2$ (Halting Condition): Nodes in $Q$ and $R$ terminate using the halting condition and output $1$.
\end{itemize}

Since the honest nodes in $P$ and $Q$ outputted different values, there is a safety violation, upon which all of the nodes send their \clues to the adjudication function $\adj$.
\Clues sent by the nodes in $Q$ and $R$ state that they did not hear from the nodes in $R$ and $Q$ respectively in step $3$ of period $1$.
By assumption, $\adj$ identifies at least one node from the set $R$ as an adversarial node.

\textbf{World 2:} This world is identical to World $1$ except that 
\begin{itemize}
    \item Nodes in $Q$ are adversarial and the nodes in $R$ are honest.
    \item Nodes in the set $Q$ behave exactly like the nodes in $R$ behaved in World $1$, i.e. the nodes in $Q$ do not send any cert-votes to the nodes in $R$ in step $3$ of period $1$ and ignore their votes at the beginning of step $4$ of period $1$. 
\end{itemize}
Thus, again the honest nodes in $P$ and $Q$ output different values, upon which all of the nodes send their \clues to $\adj$.
As worlds 1 and 2 are indistinguishable in the perspective of $\adj$, it again identifies at least one node from the set $R$ as an adversary node with non-negligible probability.
However, this is a contradiction with the definition of $\adj$ as $R$ consists of honest nodes in world 2.
\end{proof}

\section{Example Protocols Framed in Model of Section~\ref{sec:model}}
\label{sec:appendix-examples-model}

To illustrate the model in Section~\ref{sec:model}, consider a client that queries nodes running a Nakamoto-style longest chain protocol under $\AdvEnvDA$ at some time $t$. 
Suppose $\beta<1/2$.
The transcript $\tr{i}{t}$ held by node $i$ at time $t$ consists of
the blocks
received by node $i$ by time $t$.
Given the transcript $\tr{i}{t}$,
$\Wf$ outputs as \clue the longest chain implied by  $\tr{i}{t}$.
Upon collecting \clues from a subset $S$ of awake nodes with at least one honest node, a client calls $\Cf$ which selects the longest chain in the set $\{\wt{i}{t}\}_{i \in S}$ and outputs the $k$-deep prefix of that longest chain as the ledger.

We can also consider propose-and-vote-style BFT protocols such as HotStuff, LibraBFT, Streamlet and PBFT \cite{yin2018hotstuff,libraBFT,streamlet,pbft} with $n=3f+1$ nodes under $\AdvEnvPsync$.
In this case, the transcript $\tr{i}{t}$ held by a node $i$ at time $t$ consists of all received
messages such as proposals and votes.
Given $\tr{i}{t}$,
$\Wf$ outputs as \clue a sequence of proposals with votes attesting to them.
Upon collecting \clues from a subset $S$ of nodes containing at least one honest node, a client calls $\Cf$, which outputs the largest possible sequence of proposals that can be confirmed given the votes attesting to them.
The confirmation rule typically requires votes from $n-f+1$ nodes on consecutive proposals to guarantee safety which follows from a quorum intersection argument.
Liveness ensues from the fact that the honest \clue within $S$ includes all of the confirmed proposals submitted by honest nodes.
Existence of an honest \clue in $S$ is typically enforced by collecting \clues from at least $f+1$ nodes.

\end{document}